\documentclass[sigconf]{aamas} 

\renewcommand\footnotetextcopyrightpermission[1]{} % removes footnote
\usepackage{balance} % for balancing columns on the final page

\usepackage{xcolor}           % already loaded by most templates

\usepackage{amsmath}
\usepackage{mathtools}
\usepackage{amsthm}

\usepackage{multirow} % For merging cells

\usepackage{subcaption}

\newcommand{\xhdr}[1]{\vspace{5pt}\noindent\textbf{#1 }}
\newcommand{\ignore}[1]{}

\newcommand{\argmax}{\operatornamewithlimits{argmax}}

\newcommand{\var}{\text{var}}
\newcommand{\squishlist}{\begin{list}{$\bullet$}{\topsep=1pt \parsep=0pt \itemsep=1pt \leftmargin=1em }} 
\newcommand{\squishend}{\end{list}}

\newcommand{\beitemize}{\begin{list}{$\bullet$}{}} 
\newcommand{\enitemize}{\end{list}}

\theoremstyle{plain}
\newtheorem{theorem}{Theorem}

\newtheorem{lemma}{Lemma}
\newtheorem{corollary}{Corollary}
\newtheorem{assumption}{Assumption}
\theoremstyle{definition}

\theoremstyle{remark}

\newcommand{\Z}{\mathbf{Z}} % payoff vector
\newcommand{\A}{\mathcal{A}} % Allocation

\usepackage{algorithm}
\usepackage{algpseudocode}
\usepackage{float}

\setcopyright{ifaamas}
\acmConference[AAMAS '26]{Proc.\@ of the 25th International Conference
on Autonomous Agents and Multiagent Systems (AAMAS 2026)}{May 25 -- 29, 2026}
{Paphos, Cyprus}{C.~Amato, L.~Dennis, V.~Mascardi, J.~Thangarajah (eds.)}
\copyrightyear{2026}
\acmYear{2026}
\acmDOI{}
\acmPrice{}
\acmISBN{}

\acmSubmissionID{1502}

\title[GIFF]{A General Incentives-Based Framework for Fairness in Multi-agent Resource Allocation}

\author{Ashwin Kumar}
\affiliation{
  \institution{Washington University in St Louis}
  \city{St Louis, MO}
  \country{USA}}
\email{ashwinkumar@wustl.edu}

\author{William Yeoh}
\affiliation{
  \institution{Washington University in St Louis}
  \city{St Louis, MO}
  \country{USA}}
\email{wyeoh@wustl.edu}

\begin{abstract}
We introduce the General Incentives-based Framework for Fairness (GIFF), a novel approach for fair multi-agent resource allocation that infers fair decision-making from standard value functions. In resource-constrained settings, agents optimizing for efficiency often create inequitable outcomes. Our approach leverages the action-value (Q-)function to balance efficiency and fairness without requiring additional training. Specifically, our method computes a local fairness gain for each action and introduces a counterfactual advantage correction term to discourage over-allocation to already well-off agents. This approach is formalized within a centralized control setting, where an arbitrator uses the GIFF-modified Q-values to solve an allocation problem.

Empirical evaluations across diverse domains—including dynamic ridesharing, homelessness prevention, and a complex job allocation task—demonstrate that our framework consistently outperforms strong baselines and can discover far-sighted, equitable policies. The framework's effectiveness is supported by a theoretical foundation; we prove its fairness surrogate is a principled lower bound on the true fairness improvement and that its trade-off parameter offers monotonic tuning. Our findings establish GIFF as a robust and principled framework for leveraging standard reinforcement learning components to achieve more equitable outcomes in complex multi-agent systems.

\end{abstract}


\newcommand{\BibTeX}{\rm B\kern-.05em{\sc i\kern-.025em b}\kern-.08em\TeX}

\begin{document}

%%% The following commands remove the headers in your paper. For final 
%%% papers, these will be inserted during the pagination process.

\pagestyle{fancy}
\fancyhead{}

%%% The next command prints the information defined in the preamble.

\maketitle 
\sloppy
\allowdisplaybreaks

%%%%%%%%%%%%%%%%%%%%%%%%%%%%%%%%%%%%%%%%%%%%%%%%%%%%%%%%%%%%%%%%%%%%%%%%


\section{Introduction}

In many real-world applications, ranging from ridesharing~\citep{ride_alonso, shah2020neural} to allocating homelessness resources~\citep{kube2019allocating}, multiple agents vie for limited resources, often leading to significant disparities in outcomes. Traditional RL-based allocation methods focus on maximizing individual or aggregate utility, but they typically overlook fairness, resulting in inequitable resource distributions that can undermine both system-wide performance and societal acceptance. In this work, we address this gap by proposing a novel framework for fair multi-agent resource allocation that leverages the existing action-value functions to \emph{infer} fairness improvements, without any additional learning.

Existing approaches to fair RL generally incorporate fairness directly into the reward structure or modify the learning process to produce value functions that reflect fairness considerations. However, these strategies can introduce complexities such as non-stationarity and may require extensive retraining, which is impractical in many dynamic and high-stakes environments. 
Moreover, such methods typically assume that agents are intrinsically motivated to be fair, that is, they are willing to sacrifice their own utility for collective equity. In many real-world scenarios, agents may have competing objectives or operate independently without centralized coordination, making it unreasonable to expect them to adopt fairness as an inherent goal.
Our work takes a different path: Rather than embedding fairness into the value function, we infer fair decisions by analyzing the long-term Q-values computed by agents. This allows us to extract information about the potential fairness gains associated with different actions and adjust the allocation process accordingly.

The core idea of our approach, named General Incentives-based Framework for Fairness (GIFF), is to combine the standard utility-driven Q-values with additional fairness-related signals. Specifically, we decompose the fairness gain into a component that captures the marginal improvement in fairness resulting from an action, and we introduce an advantage correction term that incentivizes better-off agents to moderate their resource consumption in favor of those who are disadvantaged. This decomposition not only translates to a diverse set of fairness metrics but also offers a transparent mechanism to balance efficiency and fairness during the resource allocation process.

We formalize our framework within a centralized control setting, where a central arbitrator uses the GIFF-modified Q-values to solve an optimization problem subject to resource constraints, and wishes to incorporate fairness into the decision-making.
Through extensive experiments in domains such as ridesharing, homelessness, and job allocation, we demonstrate that our framework achieves competitive fairness-utility trade-offs compared to state-of-the-art methods. Our evaluations also highlight the critical role of the advantage correction term, particularly in environments where independent evaluations of fairness improvement fail to capture the benefits of inter-agent cooperation.

% - We introduce a new framework for fair multi-agent resource
% allocation that infers fairness from standard Q-values with-
% out altering the underlying RL models.
% = We propose a principled mechanism—comprising fairness
% gain estimation and counterfactual advantage correction—to
% adjust allocation decisions in favor of fairer outcomes.
% - We provide empirical evidence across multiple domains and
% fairness metrics to illustrate the effectiveness of our ap-
% proach. Our results are also backed by a theoretical pareto-
% improvement property of our fairness modification.

In summary, our contributions are as follows:
\begin{enumerate}
    \item We develop a \textbf{General Incentives-based Framework for Fairness (GIFF)} that can be used to improve diverse fairness metrics in multi-agent resource allocation problems without the need for additional learning. This framework infers fairness from standard Q-values without altering the underlying RL models.
    \item We propose a principled mechanism—comprising fairness gain estimation and counterfactual advantage correction—to adjust allocation decisions in favor of fairer outcomes, with only two hyperparameters.
    \item We derive instantiations of GIFF for various fairness measures like variance, $\alpha$-fairness and Generalized Gini Functions (GGF).
    \item We provide theoretical bounds on the actual fairness improvement in relation to the locally estimated fairness gain when using GIFF, for multiple fairness functions. We also provide a Pareto-improvement property over the total local fairness.
    \item We provide experimental evidence across multiple domains and fairness metrics to illustrate the effectiveness of our approach. 
\end{enumerate}

By rethinking how fairness can be integrated into multi-agent systems, our work opens up promising new directions for achieving equitable resource distribution in complex, dynamic environments---a crucial step toward deploying RL in socially sensitive applications.

\section{Related Work}
Many domains employ multi-agent resource allocation with centralized decision-makers, like ridesharing~\citep{ride_alonso, shah2020neural}, homelessness prevention~\citep{kube2019allocating, kube2023community}, satellite allocation~\citep{nauman2024satellite} and wireless networks~\citep{xue2003networks}. In our work, we focus on fair decision making in multi-agent systems that have a similar structure. There has been recent work on developing methods to learn fair policies in multi-agent RL~\citep{jiang2019FEN, siddique2020MOMDP, zimmer2021MOMDP}, using policy optimization~\citep{schulman2017PPO} with modified rewards or hierarchical policies to elicit fair behavior. \citet{mcilrathNonMarkovianFairness} look at fair resource allocation at different time horizons, learning to be fair by constructing counterfactual experiences. In our approach, we instead try to develop a learning-free method to improve fairness, by post-processing the Q-values associated with resource allocations. 
Simple Incentives~\citep{SI_kumar2023} is a related approach looking at fairness in ridesharing systems, employing domain and metric specific fairness post-processing. Our approach is much more general and adaptable to different domains and fairness metrics, as we show in our experimental results.

Beyond reinforcement learning and resource allocation, the machine learning community has broadly investigated bias and fairness. Surveys such as \citep{mehrabi2021MLBiassurvey} provide comprehensive overviews of fairness metrics and mitigation strategies, while foundational works \citep{DP_dwork2012, hardt2016equality} have laid the groundwork for understanding fairness in supervised settings. In parallel, economic and social welfare research has long provided robust formulations of equity, giving rise to measures like $\alpha$-fairness and the Generalized Gini Function—which are rooted in social welfare theory \citep{rawls1971atheory, nash1950bargaining, sen2017collective, caragiannis2019unreasonable}. By drawing on these diverse strands of literature, our work bridges multi-agent reinforcement learning, dynamic resource allocation, and fairness, contributing a novel framework that infers fairness directly from long-term Q-value estimates.

\section{Preliminaries}

We now formally describe the resource allocation problem that tackled in this paper as well as describe  several key fairness concepts that motivates our work and used later in empirical evaluations.

\subsection{Resource Allocation}
A typical resource allocation problem consists of $n$ agents $i\in\alpha$ with diverse preferences over $K$ types of resources, with a total set of resources $\mathcal{R}$ being available. Each time-step, agents attempt to take resources they want, following which new resources appear and the global state changes according to some transition dynamics.

In the centralized control setting, an arbitrator aggregates agent preferences and allocates actions while ensuring no two agents try to take the same resource. 
This constrained Multi-agent MDP model~\citep{CMMDP_de2021} is described by the tuple $\mathcal{M}$ with the following components:
\begin{align}
\mathcal{M} = \langle \alpha, S, \mathcal{O}, \{A_i\}_{i \in \alpha}, T, R, \gamma, c \rangle %incorporates the following components:
\end{align}
\squishlist
    \item $\alpha$ is the set of agents indexed by $i$ ($n$ agents).
    \item $S$ is the global state space.
    \item $\mathcal{O}: S \rightarrow O_1 \times O_2 \times \ldots \times O_n$ is the observation function that maps the true state to agent observations.
    \item $A_i$ is the action space for agent $i$, where an action $a$ includes allocation of a set of resources.
    \item $T: S \times A_1 \times A_2 \times \ldots \times A_n \times S \rightarrow [0,1]$ represents the joint transition probabilities.
    \item $R: S \times A_1 \times A_2 \times \ldots \times A_n  \rightarrow \mathbb{R}^n$ denotes the (utility) reward function, which returns a vector of rewards, one for each agent.% Note that the reward of an agent depends on the action of other agents.
    \item $\gamma$ is the discount factor for future rewards.
    \item $c: A_1 \cup A_2 \cup \ldots \cup A_n \rightarrow \mathbb{R}^K$ maps each action to its resource consumption for $K$ types of resources.
\squishend

% Each agent can then learn an action-value function $Q(o_i,a)$ that captures the long-term value of being allocated action $a$ given the observation $o_i$. 
Each agent can then learn an action-value function $Q(o_i, a)$, which estimates the expected long-term return from taking action $a$ given the local observation $o_i$, without knowledge of other agents’ current actions. Formally, this is defined as:
\begin{align}
    Q(o_i, a) = \mathbb{E} \left[ \sum_{t=0}^{\infty} \gamma^t r_i^{(t)} \,\middle|\, o_i^{(0)} = o_i, a_i^{(0)} = a \right],
\end{align}
where $r_i^{(t)}$ is the reward received by agent $i$ at time-step $t$, and $\gamma \in [0,1)$ is the discount factor. The expectation is taken over the trajectories induced by the environment dynamics and policies of all agents, conditioned only on agent $i$’s observation and action.
In any state $s$, agents compute the Q-values for all of their available actions and communicate them to the central arbitrator, who can then use them to compute a utilitarian allocation.

Let $\A$ denote the allocation of actions decided by the central allocator such that $\A_i$ is the action assigned to agent $i$.
Let $x_i(a) \in \{0,1\}$ be a binary decision variable that indicates whether agent $i$ is assigned action $a \in A_i$. Let $\mathcal{R} \in \mathbb{R}^K$ denote the vector of available resources, with $\mathcal{R}_k$ representing the quantity of resource type $k \in \{1,2,\dots,K\}$. Let $c(a) \in \mathbb{R}^K$ denote the resource consumption vector for action $a$, where $c(a)_k$ is the amount of resource $k$ consumed by action $a$. This gives us the following optimization:
\begin{align}
    \max_{x_i(a) \in \{0,1\}} \quad & \sum_{i \in \alpha} \sum_{a \in A_i} x_i(a) \cdot Q(o_i,a) \label{eq:Opt_ILP} \\
    \text{subject to} \quad 
    & \sum_{a \in A_i} x_i(a) = 1, \quad \forall i \in \alpha \label{eq:action_constraint} \\
    & \sum_{i \in \alpha} \sum_{a \in A_i} x_i(a) \cdot c(a)_k \le \mathcal{R}_k, \quad \forall k \in \{1, \ldots, K\} \label{eq:resource_constraint}
\end{align}
These constraints ensure that each agent is assigned exactly one action and that total resource usage does not exceed available supplies.
This is a general formulation, and many real-world problems follow this approach~\citep{shah2020neural, kube2019allocating, ride_alonso, nauman2024satellite, xue2003networks}. 
This problem can be formulated as an integer linear program (ILP), but more efficient algorithms and distributed approaches exist for allocation problems with stricter constraints. For instance, if the constraints boil down to a bipartite matching between agents and resources, the Hungarian algorithm~\citep{kuhn1955hungarian} can solve the allocation problem in polynomial time.

Alternatively, some systems assume agents act independently without a central arbitrator or consensus-based decision making. In this case, methods like first-come-first-served or random tie breaks are used to decide who gets contested resources. This approach is much more commonly seen in multi-agent RL~\citep{zimmer2021MOMDP, jiang2019FEN, sunehag2017VDN}. 
When using a policy optimization based approach, agents express preferences over actions which maximize their chances of getting good resources in the form of a policy $\pi$ rather than expressing valuations over bundles of resources. 

In this paper, we restrict ourselves to using Q-functions. 
Further, we select the centralized control setting, where the decision-maker has the ability to enforce fairness constraints by providing incentives to agents. Thus, in this work, we assume agents bid for actions by communicating their Q-values to the central decision maker, and each action is associated with the allocation of some resources and the corresponding gain in utility is the reward. The Q-function then captures the long-term expected utility for each agent.

% Note that resources also have an inherent value, which is distinct from the valuation of resources $u_i$ or policies $\pi$. One can imagine this as the difference between how strongly someone wants a resource, compared to the real-world cost of the resource. This can especially be important when there are long-term dynamics in the environment that the agents are in, as we will discuss next. With this in mind, we define the actual value of resources to be $R(r)$, with the shorthand $R(b) = \sum_{r\in b}R(r)$ used to denote the value of a bundle. The difference between bundles and resources will be made clear from the context.
% \xhdr{Payoff-vector ($\Z$):} 
% We are interested in resource allocation problems that have a temporal aspect to them, i.e., after each allocation, new resources may arrive in the system, and resource and agents may enter or exit the allocation pool. We consider two cases: 1) A fixed number of agents, and 2) An arbitrary number of agents that belong to a fixed number of groups. In both cases, we consider a total of $n$ groups or agents. Given this, we can construct a vector of payoffs $\textbf{Z}=[z_1, z_2, \dots, z_n]$ that captures the accumulated value of all resources allocated to each agent/group over time. This accumulated value can be additive or averaged. Our objective is to improve fairness over this payoff vector $\Z$. 

    \xhdr{Payoff-vector ($\Z$):} 
    We are interested in resource allocation problems that have a temporal aspect to them, i.e., after each allocation, new resources may arrive in the system, and resources and agents may enter or exit the allocation pool. We consider two cases: 1) A fixed number of agents, and 2) An arbitrary number of agents that belong to a fixed number of groups. In both cases, we consider a total of $n$ groups or agents. Given this, we can construct a vector of payoffs $\textbf{Z} = [z_1, z_2, \dots, z_n]$ that captures the accumulated value of all resources allocated to each agent/group over time.
    
    Unless specified otherwise, we use $z_i$ to denote the cumulative payoff for agent or group $i$ at the current time-step. When referring to a specific time, we will write $z_i^{(t)}$ to indicate the value at time-step $t$. The cumulative reward is given by:
    \begin{align}
        z_i = \sum_{\tau=0}^{t} r_i^{(\tau)},
    \end{align}
    where $r_i^{(\tau)}$ is the reward received by agent $i$ at time-step $\tau$. In some settings, an average payoff is used instead:
    \begin{align}
        \bar{z}_i = \frac{1}{t+1} \sum_{\tau=0}^{t} r_i^{(\tau)}.
    \end{align}
    
    This payoff vector serves as a temporal record of how resources have been distributed across agents and provides a foundation for incorporating fairness criteria into future allocation decisions.
    
\subsection{Fairness Concepts}

To capture fairness, we consider a fairness function 
$F: \mathbb{R}^n \rightarrow \mathbb{R},$
which maps any payoff vector $\mathbf{Z} = [z_1, \dots, z_n]$ to a numerical value such that larger values correspond to fairer distributions. In other words, for any two payoff vectors $\mathbf{Z}_1$ and $\mathbf{Z}_2$, we say that $\mathbf{Z}_1$ is fairer than $\mathbf{Z}_2$ if and only if 
$F(\mathbf{Z}_1) > F(\mathbf{Z}_2)$.

In the literature on fair multi-agent resource allocation, two main schools of thought emerge: 

% \paragraph{a. Social Welfare Function Approaches}  
% In these methods, fairness is directly embedded into a social welfare function that aggregates individual utilities~\citep{siddique2020MOMDP, zimmer2021MOMDP}. Such functions are often chosen to satisfy desirable properties like:
% \begin{itemize}
%     \item \textbf{Impartiality:} The function is invariant under any permutation of agents, ensuring equal treatment.
%     \item \textbf{Equity:} It rewards reallocations that transfer resources from better-off agents to worse-off agents (consistent with the Pigou-Dalton principle).
%     \item \textbf{Efficiency:} If every agent receives higher or equal utility in one allocation over another, the function assigns a higher value to the preferable allocation.
% \end{itemize}
\squishlist
\item \textbf{Social Welfare Function Approaches.} In these methods, fairness is directly embedded into a social welfare function that aggregates individual utilities~\citep{siddique2020MOMDP, zimmer2021MOMDP}. Such functions may be designed to satisfy three desirable properties: They exhibit \textbf{impartiality}, remaining invariant under any permutation of agents to ensure equal treatment; they promote \textbf{equity} by rewarding reallocations that transfer resources from better-off agents to those worse-off, consistent with the Pigou-Dalton principle; and they ensure \textbf{efficiency} by assigning a higher value to allocations in which every agent receives higher or equal utility compared to alternative allocations.

These metrics usually take the form of a summation over transformations of agent utilities, for example:
\begin{itemize}
\item \textbf{$\alpha$-Fairness:}
For a given payoff vector $\mathbf{Z} = [z_1, \dots, z_n]$, define the per-agent $\alpha$-fair utility as:
\begin{align}
U_\alpha(z) &= \begin{cases}\frac{z^{\,1-\alpha}}{1-\alpha} & \text{if } \alpha \neq 1,\\ \log(z) & \text{if } \alpha = 1,\end{cases} \label{eq:ALF} 
\end{align}
and the overall fairness measure as:
$F_{\alpha}(\mathbf{Z}) = \sum_{i=1}^n U_\alpha(z_i)$.

\item \textbf{Generalized Gini Function (GGF):}
Order the components of $\mathbf{Z}$ as $z_{(1)} \leq z_{(2)} \leq \cdots \leq z_{(n)}$. Then, the GGF function is defined as
\begin{align}
    F_{GGF}(\mathbf{Z}) = \sum_{i=1}^n w_i\, z_{(i)}, \label{eq:GGF} 
\end{align}
with weights satisfying $w_1 \ge w_2 \ge \cdots \ge w_n$ and $\sum_{i=1}^n w_i = 1$.
\end{itemize}
By varying $\alpha$ or $w_i$, these metrics can transition between utilitarian and egalitarian fairness. Specifically, $\alpha$-fairness with $\alpha=1$ is equivalent to the popular log Nash Welfare metric.

% \memo{Are there more papers we can cite below aside from our paper?}

\item \textbf{Distributional Approaches.}  
Alternatively, fairness may be measured separately via distributional metrics—such as variance, the Gini index, or Jain's fairness index—and then combined with total utility~\citep{SI_kumar2023, raman2020ridefair}. These metrics can capture non-linear relationships among agents' utilities and provide a distinct measure of fairness that is later interpolated with the overall efficiency.
\squishend

Having outlined these approaches, our work aims to balance efficiency and fairness via a joint objective. Specifically, for a given time horizon $T$ with payoff vector $\mathbf{Z}_T$, we seek an allocation policy that maximizes:
\begin{align}
\max\quad (1-\beta) &U_T + \beta F(\mathbf{Z}_T), \label{eq:Objective}\\
U_T &= \sum_{i=1}^{n} z_i,
\end{align}
Here, $\beta \in [0,1]$ is a trade-off parameter that controls the relative importance of efficiency (total utility) versus fairness.

In our evaluations, we show results using both kinds of fairness metrics. In particular, for the social welfare function approaches, we incorporate a specialized term that enables agents to locally assess the benefits of their actions for others. Notably, our incentives-based framework operates entirely online with zero additional training.

\section{GIFF}
The key insight in our work is that the Q-function, which captures the long-term utilitarian effects of actions, can also be used to guide the decision maker towards a fair allocation. Previous work has considered directly learning to optimize for fairness~\citep{zimmer2021MOMDP, jiang2019FEN} or using knowledge of the true reward function to myopically improve variance using domain-specific post-processing~\citep{SI_kumar2023}. Instead, we provide a General Incentives-based Framework for Fairness (GIFF) that improves fairness without the need for additional learning for a variety of domains and fairness functions.

To achieve this, we develop an approach that takes advantage of the instantaneous pre-decision payoffs $\textbf{Z}_t$ at the current time $t$ to guide the present decision towards a fairer outcome by improving the perceived value of fairer allocations. We do this by computing the estimated improvement to fairness, the \emph{fairness gain} of actions, and augmenting the pre-trained Q-values to reflect our objective (Eq.~\ref{eq:Objective}).
% In our work, we aim to provide a method to improve fairness in a dynamic resource allocation problem given that an efficient solution is already computable. We derive methods to myopically optimize diverse fairness functions to improve the allocations by modifying the agents' decisions to align them better with the fairness objective. Inspired by prior work~\cite{SI_kumar2023}, we derive a General Incentive-based Fairness Framework (GIFF).

% Given a system that evolves over time, let $\textbf{Z}_T$ denote the payoff vector at a time horizon $T$. Then , our goal with GIFF is to maximize $F(\textbf{Z}_T)$. 

% \subsection{The Ideal Scenario}
First, we assume an idealized scenario.
Let $\A=[a_i]_{i\in\alpha}$ denote an allocation that contains one action for each agent. We overload the notation for the reward function to let $R(\A_i)$ be a shorthand for the true reward received by agent $i$ under allocation $\A$. The updated payoff vector $\Z_{t+1}|\A$ can be computed using $\Z_t$ and $\A$, by updating the payoffs using the true rewards. Then, the fairness gain for any allocation can be defined as:
\begin{align}
    \Delta F(\A) &= F(\textbf{Z}_{t+1}) - F(\textbf{Z}_t)\\
    z_i^{t+1} &= z_i^{t} + R(\A_i) & \forall i \\
    \A^*_f &= \argmax_\A  \Delta F(\A)
\end{align}

Here, $\A^*_f$ is the allocation that improves fairness the most in the current step. We can also conceive of a similar search which maximizes over an entire sequence of allocations to maximize long-term fairness. However, 
% even the one-step optimal allocation is very hard to compute as it involves a search over all joint actions, which can quickly become intractable. 
even for the one-step allocation, the search space is combinatorial in the number of agents and their respective action spaces, as we have to consider all possible joint actions. This makes the global optimization intractable, necessitating alternate methods for computing a fair allocation.

\subsection{Using Q-values to Estimate Fairness}
We observe that it is much easier to reason about fair actions if we can decompose the fairness gain across agents. To achieve this, we reason only over the locally conditioned updated payoff vector $\Z^{a^i}_t$, updating all accumulated utilities based only on a single agent's action $a^i$, keeping everything else unaffected.
\begin{align}
    z_j^{t+1} &= z_j^t + \mathbb{I}\{j=i\}\, R(a^i) & \forall j
\end{align}

% Additionally, since we may not have access to the true reward function $R(\A)$, we use the Q-function to update the payoff:
In many real problems, having access to the true reward function $R(a)$ is unlikely. Further, in dynamic environments, agents may take a critical action earlier, which leads to a payoff after multiple steps; however, a critical decision towards getting to it may happen much earlier. 
% Consider an environment where a narrow passage leads to a resource, and the only action is to move forward. In this setting, the agent locks in the reward when it enters the passage, and the state immediately before collecting the resource holds little importance, requiring us to find better ways of approximating long-term fairness returns.
% 
If the agent is a Q-learner, we can leverage the fact that the Q-values encode the long-term value of taking certain actions, and the difference in Q-values will be small in states where all routes lead to similar payoffs. Thus, if we use the Q-value as a proxy for the reward function, we can account for long-term returns without knowing the reward function or the environment dynamics.
\begin{align}
    \Delta F(a^i) &= F(\Z^{a^i}_{t+1}) - F(\Z_t) \label{eq:fair_reward}\\
    z_j^{t+1} &= z_j^t + \mathbb{I}\{j=i\}\, Q(o_i, a^i) & \forall j \label{eq:z_local_update}
\end{align}

This quantity $\Delta F(a^i)$, termed the \emph{fairness gain} measures the marginal (local) impact of agent $i$'s action on the overall fairness of the allocation given the current distribution of resources $\Z_t$.
Note that computing this for all agent actions is linear in the size of the action space for each agent.
This has two benefits. First, we do not need to have access to the true reward function (which is not available in many cases) and, second, we capture more than just the immediate return, allowing us to capture long-term effects of certain allocations. However, this has the drawback that it does not capture inter-agent interactions very well. Thus, we introduce an additional mechanism that can provide this information.

% \subsection{Counterfactual Advantage Correction}
\subsection{Advantage Correction: Incorporating Counterfactual Fairness Gains}

Fairness concerns in dynamic resource allocation may also require that agents who are already well-off (i.e., have a high accumulated utility) are discouraged from taking resources that can help worse-off agents improve their return, thus allowing disadvantaged agents to catch up. This is also a desirable property of social welfare functions, termed \textbf{equity}, used to capture the notion of accumulated wealth: Moving resources from better-off agents to worse-off agents should improve the fairness function's value. This is also known as the Pigou-Dalton principle~\citep{zimmer2021MOMDP} and has been discussed in previous works in fair multi-agent RL.

However, in practice, the local fairness gain $\Delta F(a^i)$---which reflects only the immediate change in an agent's own utility---does not capture the broader altruistic impact of reallocating resources. In many fairness metrics, this counterfactual update considers solely the acting agent's payoff, thereby ignoring the significant improvement in fairness that would occur if the resource were allocated to a disadvantaged agent. Consequently, even when an agent’s decision to forgo an action yields $\Delta F=0$, it may still lead to substantial overall fairness gains if the resource were reallocated. This limitation motivates the inclusion of a method to more accurately account for these counterfactual benefits.

To this end, we introduce an \emph{advantage correction} term that incentivizes better-off agents to give up their top preferences to benefit other disadvantaged agents.
Recall that in our setting, when agent $i$ takes action $a$, its local fairness gain $\Delta F(a)$ is computed by updating the agent's payoff $z_i$ with $Q(o_i, a)$ (Eqs.~\ref{eq:fair_reward} and~\ref{eq:z_local_update}).
To measure the counterfactual benefit of this action, we consider allocating this resource to any other agent that can take this action $j\neq i, a\in A_j$, and compute the counterfactual benefit $\Delta F^{(j)}$, using $Q(o_j, a)$ to update agent $j$'s payoff:
\begin{align}
\Delta F^{(j)} &= F\bigl(\mathbf{Z}^{(j)}_{t+1}\bigr) - F(\mathbf{Z}_t),\\
z_j^{t+1} &= z_j^t + Q(o_j, a)
\end{align}
% where $\mathbf{Z}^{(j)}_{t+1}$ is obtained by updating only $z_j$. 
Let us define the set of these candidate counterfactual agents as $\alpha_c(a) = \{j\in\alpha:j\neq i, a\in A_j\}$.
We can then compute the average counterfactual fairness improvement:
\begin{align}
\Delta F_{\text{avg}}(a) = \frac{1}{|\alpha_c(a)|}\sum_{j\in\alpha_c(a)} \Delta F^{(j)}
\end{align}

To capture the benefit of allocating this resource to agent $i$, we can calculate the advantage function:
\begin{align}
F_{\text{adv}}(a) = \Delta F(a) - \Delta F_{\text{avg}}(a)
\end{align}
A negative $F_{\text{adv}}$ suggests that another agent would benefit more from the resource allocation, whereas a positive $F_{\text{adv}}$ indicates that the current agent can better improve fairness. If the fairness metric follows the principle of equity, then we expect agents with higher fairness gain to be the disadvantaged agents.
Instead of $\Delta F_{\text{avg}}$, alternate baselines like the maximum fairness improvement may also be considered.

To integrate this counterfactual fairness measure with the action's inherent quality, we weigh the fairness advantage by the relative Q-value gap:
\begin{align}
\Delta Q(a) = Q(o_i,a) - \min_{a'\in A_i} Q(o_i,a'),
\end{align}
which reflects how much better action $a$ is compared to the worst option for agent $i$. This is helpful in preventing disproportional changes because of Q-value overestimation.

Finally, we define the counterfactual advantage correction term as:
\begin{align}
\Delta Q_{\text{adv}}(a) = F_{\text{adv}}(a)\, \Delta Q(a).
\label{eq:adv_correction}
\end{align}
This formulation has the following intuitive implications:
\squishlist
    \item If the fairness gain $\Delta F(a)$ is lower than the mean counterfactual gain $\Delta F_{\text{avg}}(a)$, then $F_{\text{adv}}(a)$ is negative, leading to a negative $\Delta Q_{\text{adv}}(a)$. This reduces the attractiveness of action $a$, discouraging further accumulation by already advantaged agents.
    \item Conversely, if $\Delta F(a)$ exceeds $\Delta F_{\text{avg}}(a)$, then $F_{\text{adv}}(a)$ is positive, and $\Delta Q_{\text{adv}}(a)$ is positive. This boosts the value of actions that help a disadvantaged agent catch up.
\squishend

\subsection{GIFF-modified Q-values}
We combine the original Q-value estimate with the fairness gain and the counterfactual advantage correction to obtain the GIFF-modified Q-value, which we can use in the optimization in Equation~\ref{eq:Opt_ILP} to compute allocations:
\begin{align}
    % Q_f(a) = \Bigl[\,\Delta F(a) + \delta\, \Delta Q_{\text{adv}}(a)\Bigr]  \nonumber
    Q_f(a) = \Delta F(a) + \delta\, \Delta Q_{\text{adv}}(a) \nonumber
\end{align}
\begin{align}
    \boxed{
    Q^{\text{GIFF}}(o_i, a, \beta, \delta) = (1-\beta)\, Q(o_i,a) + \beta\, Q_f(a)
    \label{eq:GIFF_modification}
    }
\end{align}
% where:
\squishlist
    \item \( \beta \in [0,1] \) controls the trade-off between efficiency (standard Q-values) and fairness, with $\beta=1$ leading to allocations based purely off the fairness gain.
    \item $\delta \ge 0$ controls the degree of advantage correction. Empirically, we observed that small positive values~($<0.5)$ lead to good results, but this is dependent on the environment and fairness function being used.
\squishend
In practice, by adjusting \( \beta \) and \(\delta\), the central optimizer can be nudged toward actions that balance both overall system utility and fairness, as measured by \( F(\Z) \).

\xhdr{Approximation for Distributed Computation:}
Computing the counterfactual fairness gains requires access to the Q-values of all agents. However, in many practical settings, collecting the true Q-values from every agent at each time-step may be infeasible due to communication constraints or scalability concerns.

To reduce this overhead, we adopt an approximation where each agent assumes that all other agents evaluate actions similarly to themselves. That is, agent $i$ approximates the Q-values of other agents $j \in \alpha_c$ as:
\begin{align}
    Q(o_i, a) \simeq Q(o_j, a) \quad \forall j \in \alpha_c,
\end{align}
where $\alpha_c \subseteq \alpha$ is the subset of agents competing for $a$.

This assumption enables each agent to use their own Q-value estimates as stand-ins for those of others when computing fairness-aware updates. As a result, the modified Q-values can be computed using only the current utility vector $\mathbf{Z}_t$ and the agent's local Q-values, avoiding the need for global Q-value communication.

% \subsection{Approximation for Tractable Computation}

% In multi-agent settings, computing counterfactual fairness gains for all agents would require access to every agent's true Q-values. To avoid this, we make the approximation
% \begin{align}
% Q(o_i,a) \simeq Q(o_j,a) \quad \text{for all } j.
% \end{align}
% This assumption implies that the long-term value of an action is roughly similar across agents, which allows us to use the same set of predicted payoffs $Q(o_i,a)$ for estimating both the direct fairness gain $\Delta F^{(i)}$ and the counterfactual gains $\Delta F^{(j)}$. Consequently, the modified Q-values can be computed solely based on the current utility vector $\mathbf{Z}$ and local Q-value estimates, without requiring the true Q-values of all other agents.

% This derivation lays the theoretical foundation for our new advantage correction mechanism in GIFF, ensuring that actions are evaluated not only by their immediate or long-term returns but also by their fairness implications across agents.

\section{Theoretical Results}
\label{sec:giff_theory_results}

In this section, we provide theoretical guarantees for the core mechanism of GIFF. The algorithm's fairness-aware Q-value, $Q^{\text{GIFF}}$, relies on a tractable surrogate for the true, combinatorial fairness improvement of a joint allocation. Specifically, GIFF approximates the true joint gain by summing the individual, local fairness gains of each agent's action (or Q-value), a quantity we formally define as the surrogate, $S$. We now prove three key properties of this design: (1) This surrogate $S$ is a principled lower bound on the true fairness improvement for several canonical fairness functions; (2) Maximizing this surrogate is guaranteed to improve as the fairness weight $\beta$ increases; and (3) How these guarantees on our surrogate translate to guarantees on the realized, real-world fairness.

Throughout, let $\Z=(z_1,\dots,z_n)\in\mathbb{R}^n$ be the current payoff vector and $y=(y_1,\dots,y_n)\in\mathbb{R}^n_{\ge 0}$ be the vector of utility increments from a feasible allocation. For a fairness function $F:\mathbb{R}^n\to\mathbb{R}$, we define the \emph{realized fairness improvement} as $\Delta_{\mathrm{joint}} := F(\Z+y)-F(\Z)$ and the \emph{local fairness gain} for agent $i$ as $\Delta^{\mathrm{local}}_i := F(\Z+y_i e_i)-F(\Z)$, where $e_i$ is the $i$-th unit vector. GIFF's objective uses the sum of local gains, $S := \sum_{i=1}^n \Delta^{\mathrm{local}}_i$, as a surrogate for the true joint improvement.

Our proofs rely on the following assumptions:
\begin{assumption}[Nonnegative increments]
\label{asmp:nonneg-y}
Any agent's change in utility from an allocation is nonnegative: $y_i \ge 0$ for all $i$.
\end{assumption}
\begin{assumption}[Q-value correctness]
\label{asmp:Q-corr}
Q-values are perfectly accurate predictors of utility increments.
\end{assumption}

\subsection{Local–Gain Lower Bound}
\label{subsec:theory-lg-lb}

Our first result establishes that the surrogate $S$ is a conservative lower bound on the realized fairness improvement $\Delta_{\mathrm{joint}}$ for four canonical fairness metrics.

\begin{theorem}[Local–Gain Lower Bound]
\label{thm:local-gain-lb}
Let $\Z\in\mathbb{R}^n$ be a payoff vector and $y\in\mathbb{R}^n_{\ge 0}$ be a nonnegative increment vector. For each fairness function $F\in \{F_\alpha, F_{GGF}, F_{var}, F_{min}\}$, the realized joint gain dominates the sum of local gains:
$$
\Delta_{\mathrm{joint}} \;\ge\; S = \sum_{i=1}^n \Delta^{\mathrm{local}}_i,
$$
% under the stated assumptions:
% \begin{enumerate}
%     \item $F=F_\alpha$ ($\alpha$-fairness): $\alpha\ge 0$ and utilities lie in the domain of $U_\alpha$.
%     \item $F=F_{\mathrm{var}}$ (negative variance): no further assumptions.
%     \item $F=F_{\mathrm{GGF}}$ (generalized Gini) with nonincreasing weights $w_1\ge\cdots\ge w_n$.
%     \item $F=F_{\min}$ (maximin): no further assumptions.
% \end{enumerate}
% Moreover, equality holds for $\alpha$-fairness. For the other metrics, equality holds only under specific conditions (e.g., when at most one agent receives a positive increment for negative variance).
\end{theorem}

\begin{proof}[Proof Sketch]
The result for $\alpha$-fairness follows from the separability of the function, which yields exact equality. For negative variance, the inequality arises from the non-negative cross-term $\frac{2}{n^2}\sum_{i<j} y_i y_j$ in the expression for $\Delta_{\mathrm{joint}} - S$. The proofs for GGF and maximin rely on the properties of order statistics and case analysis on the minimum-achieving agents, respectively. Full proofs are provided in the appendix.
\end{proof}

\subsection{Monotonicity of Surrogate Fairness in \texorpdfstring{$\beta$}{beta}}
\label{subsec:theory-monotone}

Next, we show that GIFF's surrogate objective is guaranteed to be nondecreasing as the fairness weight $\beta$ increases. This provides a reliable mechanism for tuning the fairness-utility trade-off.

\begin{theorem}[Monotone Surrogate Fairness]
\label{thm:monotone-surrogate}
Fix a decision round and let $\mathcal A$ be the finite set of feasible joint allocations. Let $U(A)$ be the total utility and $S(A)$ be the sum of local fairness gains for an allocation $A \in \mathcal A$. Let $A^*(\beta)$ be the allocation chosen by GIFF for a given fairness weight $\beta \in [0, 1)$. For any $0 \le \beta_1 < \beta_2 < 1$, the corresponding surrogate fairness values are ordered:
$$
S(A^*(\beta_2)) \;\ge\; S(A^*(\beta_1)).
$$
\end{theorem}

\begin{proof}[Proof Sketch]
Let $\theta = \beta/(1-\beta)$. The GIFF objective is to maximize $G_\theta(A) = U(A) + \theta S(A)$. By optimality, we have $G_{\theta_1}(A_1) \ge G_{\theta_1}(A_2)$ and $G_{\theta_2}(A_2) \ge G_{\theta_2}(A_1)$. Subtracting these two inequalities yields $(\theta_2-\theta_1)(S(A_2)-S(A_1)) \ge 0$. Since $\theta_2 > \theta_1$, we must have $S(A_2) \ge S(A_1)$.
\end{proof}

\begin{corollary}[Strict Increase at a Switch]
\label{cor:strict-switch}
If the maximizer of the GIFF objective is unique at $\beta_1$ and $\beta_2$, and $A^*(\beta_1) \neq A^*(\beta_2)$, then $S(A^*(\beta_2)) > S(A^*(\beta_1))$.
\end{corollary}

\subsection{From Surrogate to Realized Fairness Guarantees}
\label{subsec:theory-slack}

Finally, we connect the surrogate guarantees to realized fairness by bounding the \emph{slack}, defined as $\mathrm{slack} := \Delta_{\mathrm{joint}} - S$. By Theorem~\ref{thm:local-gain-lb}, the slack is always non-negative for the metrics considered. We can derive exact expressions or tight bounds for it:
\squishlist
    \item \textbf{$\alpha$-fairness:} The surrogate is exact, so $\mathrm{slack} = 0$.
    \item \textbf{Negative Variance:} The slack is a precisely computable quadratic term, $\mathrm{slack} = \frac{2}{n^2}\sum_{i<j} y_i y_j$.
    \item \textbf{GGF \& Maximin:} The slack depends on whether utility increments cause agents to change ranks (for GGF) or on the uniqueness of the minimum-payoff agent (for maximin).
\squishend

These bounds allow us to translate the monotonicity of the surrogate into a guarantee on realized fairness. When an increase in $\beta$ triggers a switch to a new allocation, if the surrogate $S$ increases by more than the maximum possible slack of the previous allocation, the realized fairness $\Delta_{\mathrm{joint}}$ is guaranteed to strictly increase. For $\alpha$-fairness, any switch to a different allocation with a higher surrogate value implies a strict increase in realized fairness.

In summary, our theoretical results provide a firm foundation for GIFF. We have shown that its core mechanism of using a sum-of-local-gains surrogate is principled (Theorem~\ref{thm:local-gain-lb}), predictable (Theorem~\ref{thm:monotone-surrogate}), and directly translatable into guarantees on realized fairness (Section~\ref{subsec:theory-slack}). This analysis confirms that GIFF is not just a heuristic but a structured framework for improving fairness in complex allocation problems.

\section{Empirical Results}

We show results from two experiments. First, we compare GIFF to an existing domain-specific method for ridesharing, optimizing variance (a distributional metric). Then, we extend this to a new domain of homelessness, minimizing the Gini index. Finally, we look at SWF-based metrics in a domain highlighting the need for the counterfactual advantage correction.

% \memo{A possible experiment: comparing the quality and time of global optimization (combinatorial search) with GIFF modified optimization.}

\begin{figure}[t]
    \centering
    \includegraphics[width=0.48\linewidth]{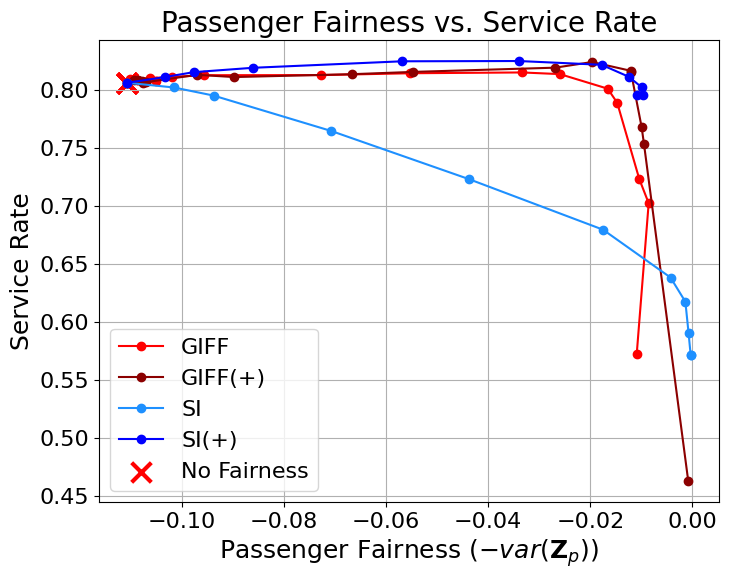}
    \includegraphics[width=0.48\linewidth]{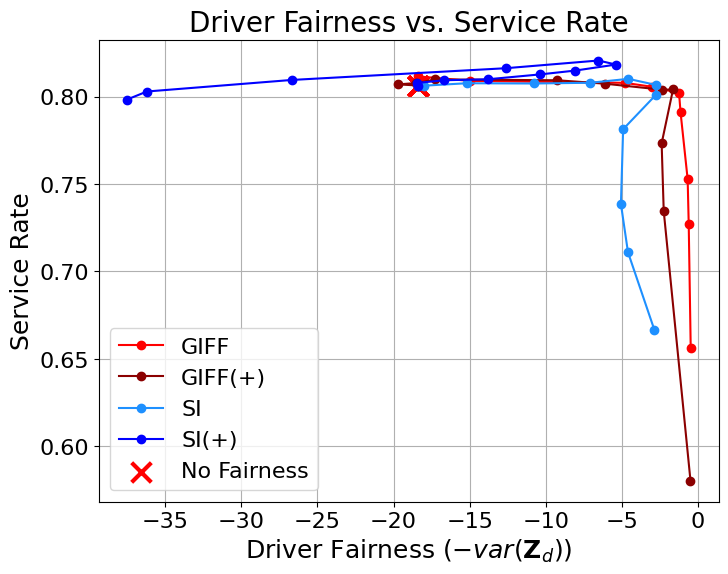}
    \caption{Comparison of fairness versus system utility. Each line is plotted in order of increasing fairness tradeoff weight $\beta$, starting from the red X ($\beta=0$) Top: Passenger fairness (measured as $-var(\mathbf{Z}_p)$) versus overall service rate. Bottom: Driver fairness (measured as $-var(\mathbf{Z}_d)$) versus overall service rate. The red X indicates the baseline performance (raw Q-values without fairness adjustments). 
    % GIFF consistently improves the fairness-utilty tradeoff over SI. While SI(+) provides a marginal improvement for passengers, incorporating the same heuristic into GIFF (GIFF(+)) yields performance on par with SI(+). For drivers, however, SI(+) starts to degrade fairness at high fairness weights, whereas GIFF maintains a stable fairness--efficiency tradeoff.
    }
    \label{fig:SI_results}
\end{figure}

\subsection{Baseline Comparisons in Real-World Domains: Ridesharing}
We test our approach in the complicated ridesharing domain~\citep{shah2020neural, SI_kumar2023}, where passengers are allocated to drivers in a dynamic matching environment. Part of the complexity also arises from the fact that more than one passenger can be allocated to the same driver, sharing the trip, leading to a huge combinatorial search space that is difficult to directly optimize. Each vehicle estimates the Q-values based on groups of passengers, and the central allocation maximizes this Q-value for all drivers, subject to passenger constraints.

We compare our results to SI~\citep{SI_kumar2023}, a recent approach for myopic fairness designed for the ridesharing application. SI's fairness objective is to reduce variance in groups of passengers (measured in terms of service rate) and in drivers (measured as differences in trips assigned). Their approach operates in a centrally constrained resource allocation environment by augmenting the base model with an additive fairness term. In our experiments, we compare both the original SI and its heuristic variant SI(+), which clips negative fairness incentives to focus solely on improvements. We also implement a corresponding heuristic in our method, denoted GIFF(+).

Both GIFF and SI start from the same base model that predicts raw Q-values (indicated by the red X's in Figure~\ref{fig:SI_results}, which represents the baseline with no fairness adjustments). Our goal is to demonstrate that GIFF not only improves fairness over SI, but also avoids the drawbacks of SI(+) at high fairness weights.
For passengers, fairness is measured as the variance in service rate for groups traveling between source and destination neighborhoods ($\mathbf{Z}_p$). For drivers, the objective is to minimize variance in driver income, measured by the number of trips per driver ($\mathbf{Z}_d$). In our simulation, 1000 vehicles are deployed on the island of Manhattan, using a real-world dataset from New York City, capturing passenger requests between 8am and 12pm during the busy morning hours.

Our results, shown in Figure~\ref{fig:SI_results}, demonstrate that GIFF consistently achieves a better fairness-utility tradeoff than SI for both passengers and drivers.
For passengers, GIFF outperforms the base SI method. While the heuristic variant SI(+) provides a marginal improvement, applying the same heuristic to our method (GIFF(+)) matches its performance. This shows that GIFF's core formulation is strong and can be enhanced with simple heuristics when long-term values are unavailable.

The superiority of GIFF is most pronounced for drivers. As the fairness weight $\beta$ increases, GIFF maintains a stable and favorable tradeoff. In contrast, SI(+) becomes counterproductive, eventually degrading fairness to a level \textbf{worse than the baseline} with no fairness adjustments (indicated by the red X). This instability at high fairness weights highlights a critical drawback of the SI(+) heuristic, whereas our method, GIFF, remains robust and effective across the full range of fairness weights.

\xhdr{Fairness with changing $\beta$:}
In SI, the objective is additive ($U +\beta F$), as opposed to GIFF's weighted combination (Eq.~\ref{eq:Objective}). To compare the effect of this tradeoff weight on fairness, we transform the weights for GIFF as $\frac{\beta}{1+\beta}$, and plot the change in fairness for both SI and GIFF with this hyperparameter on a logarithmic scale.  
Figure~\ref{fig:variance_analysis} shows the tradeoff between fairness weight and the variance in utilities $\Z_p$ and $\Z_d$. 
 % To properly show the hyperparameter search space in these plots, the GIFF method uses the transformation $\log\left(\frac{\beta}{1-\beta}\right)$, while the SI method uses $\log(1+\beta)$.

\begin{figure}[t]
    \centering
    \subfloat[Passenger Fairness]{\includegraphics[width=0.49\linewidth]{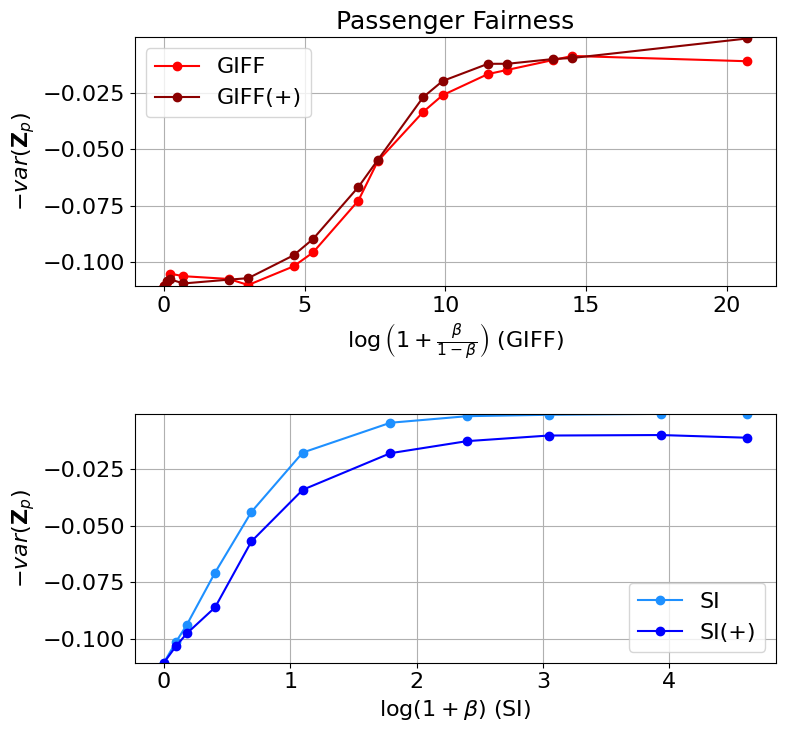}}%\quad
    \subfloat[Driver Fairness]{\includegraphics[width=0.49\linewidth]{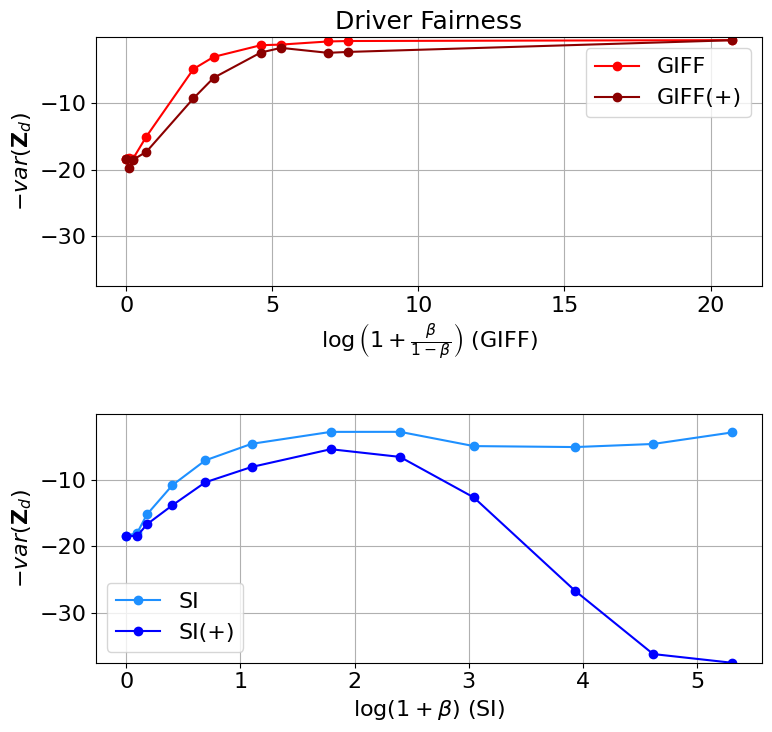}}
    \caption{Variance vs. Fairness Weight. Top row: GIFF results using $\log\left(1+\frac{\beta}{1-\beta}\right)$. Bottom row: SI results using $\log(1+\beta)$. 
    % Left: Passenger fairness ($-var(\mathbf{Z}_p)$). Right: Driver fairness ($-var(\mathbf{Z}_d)$).
    }
    \label{fig:variance_analysis}
\end{figure}

% \noindent \textbf{Passenger Analysis:} In the top left panel, as the fairness weight increases, GIFF steadily reduces the variance in passenger utilities, indicating better equity. In contrast, SI (bottom left) shows an initial reduction that levels off at higher fairness weights.
For passengers, both GIFF and SI keep improving fairness as $\beta$ is increased, as shown in the left panels. 
In the top right panel, GIFF continues to lower the variance for drivers as the fairness weight increases. However, SI(+) (bottom right) eventually worsens fairness, even falling below the baseline performance. This demonstrates that GIFF achieves a more stable fairness--utilty tradeoff.

\begin{figure}
    \centering
    \includegraphics[width=0.99\linewidth]{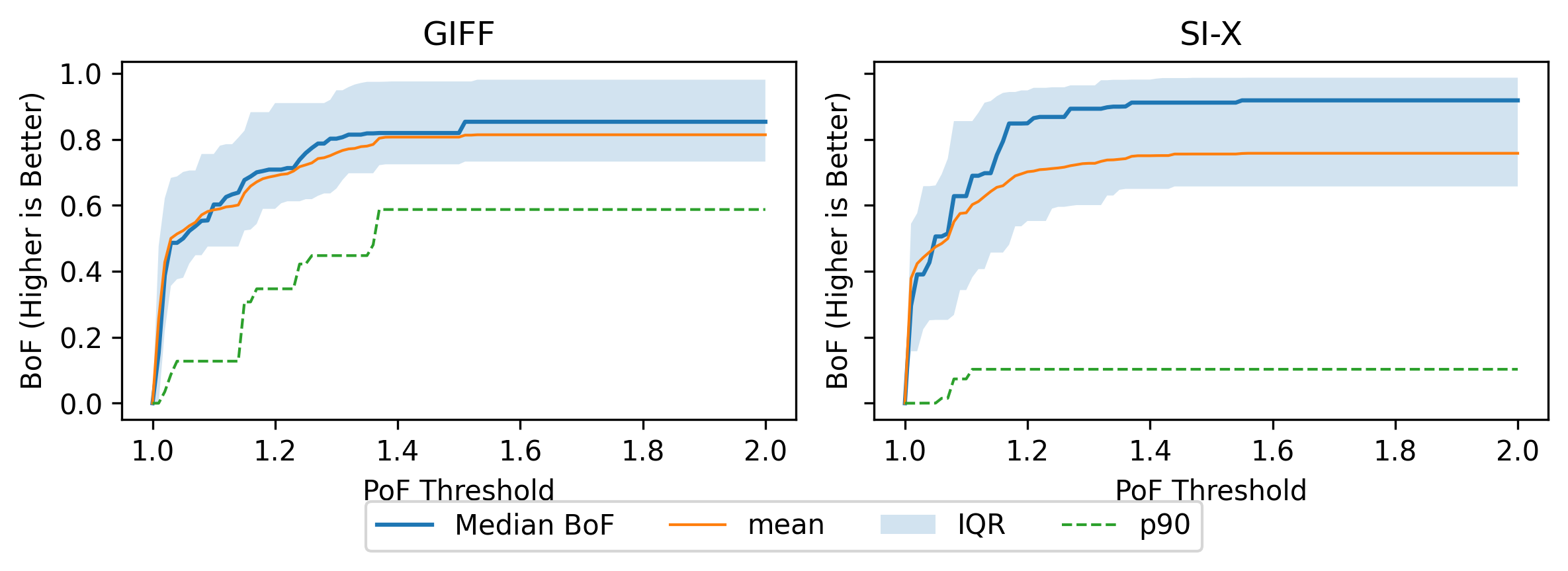}
    \includegraphics[width=0.99\linewidth]{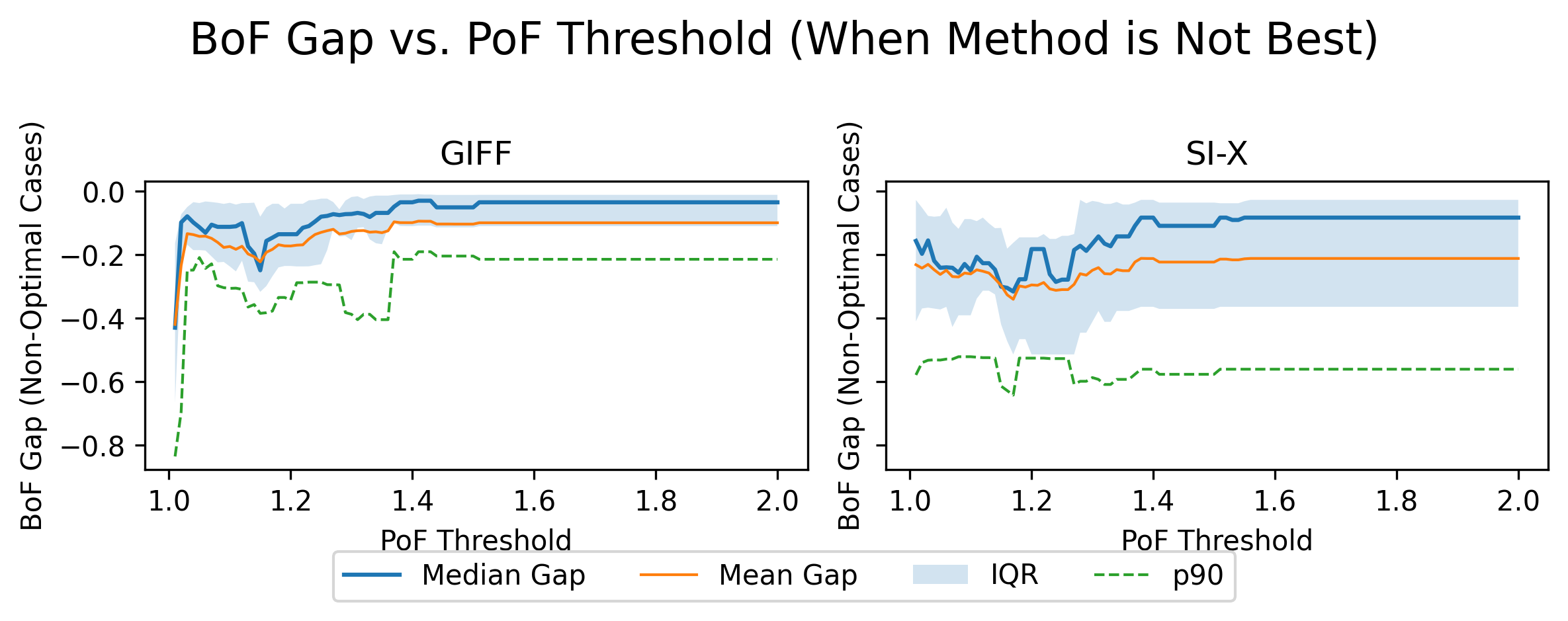}
    \caption{Results for the homelessness dataset. Top: Comparison of the benefit of fairness (BoF) distribution as the price of fairness (PoF) threshold is increased. Bottom: The BoF gap compared to the best method, excluding BoF=0. Each vertical slice corresponds to the distribution over all 38 features.}
    \label{fig:homelessness_results}
\end{figure}

\subsection{Generalization to Homelessness Prevention}
To demonstrate the versatility of our approach beyond dynamic, Q-value-based environments, we test GIFF in the domain of homelessness prevention using a real-world dataset~\cite{kube2019allocating}. Here, the task is to assign households to one of four interventions to minimize the total probability of re-entry into homelessness. The ``cost'' of assigning household $h$ to intervention $a$ is given by a pre-computed counterfactual probability, $\Pr(h,a)$.

Our framework is adapted to this new context by treating it as a cost-minimization problem. We use the negative of the re-entry probabilities as utility values, so that $Q^{\text{GIFF}}(h, a, \beta, \delta) = (1-\beta)\, (-\Pr(h,a)) + \beta\, Q_f(a)$. To further highlight GIFF's flexibility, we move beyond variance reduction and adopt the \textbf{Gini coefficient} as the fairness metric, $F_{\mathrm{gini}} = - \frac{\sum_{i=1}^n \sum_{j=1}^n |z_i - z_j|}{2n \sum_{k=1}^n z_k}$, where $z_i \in [0,1]$ is the average re-entry probability for a given demographic group.

The dataset includes 38 household features (e.g., race, gender, family size), and we run 38 independent experiments, each defining fairness groups based on one feature. This allows us to assess the robustness of each method across a wide variety of fairness definitions. Since the original SI method is not applicable directly, we developed a competitive baseline, Simple Incentives - Extended (SI-X). Further details about the dataset, implementation and SI-X formulation are included in the supplement.

To evaluate performance across these 38 experiments, we introduce two metrics:
\squishlist
    \item \textbf{Price of Fairness (PoF):} The ratio of the total re-entry probability with fairness adjustments to the baseline (fairness-unaware) total probability. A PoF of 1.05 means a 5\% increase in the overall re-entry rate.
    \item \textbf{Benefit of Fairness (BoF):} The percentage reduction in the Gini coefficient compared to the baseline, calculated as $1 - \frac{\text{Gini}(\text{new})}{\text{Gini}(\text{base})}$.
\squishend

Figure~\ref{fig:homelessness_results} summarizes the results by plotting the distribution of BoF achieved for a given PoF threshold. Each vertical slice represents the BoF distribution over all 38 feature-based groupings.
The top panel shows that GIFF is a more effective and reliable method. On 90\% of the features, GIFF is able to get 60\% improvement in $F_{gini}$ compared to the baseline. GIFF consistently yields a higher \textbf{mean} BoF. More importantly, GIFF demonstrates superior worst-case performance; its 90th percentile is substantially higher than that of SI-X, indicating that GIFF avoids the severe fairness failures that SI-X is prone to on certain groups. 
% The higher 25th and 75th percentiles for GIFF also confirm its ability to achieve greater fairness improvements when conditions are favorable.
The bottom panel reinforces this conclusion by analyzing the \textbf{BoF Gap}, which measures how much a method underperforms \textit{only on the tasks where it was not the best}. The gap for GIFF is minimal and concentrated near zero, showing that even when it is not the top performer, it is a close second. In contrast, SI-X exhibits a wide gap distribution, with its 90th percentile exceeding 0.4. This means that when SI-X fails, it fails badly, achieving fairness outcomes that are dramatically worse than what is possible with GIFF.

Overall, these results in a distinct problem domain with a non-linear fairness metric confirm that GIFF is a robust and broadly applicable framework for integrating fairness into resource allocation systems.

\begin{figure}[t]
    \centering
    \begin{subfigure}{1\linewidth}
        \centering
        \includegraphics[width=0.495\linewidth]{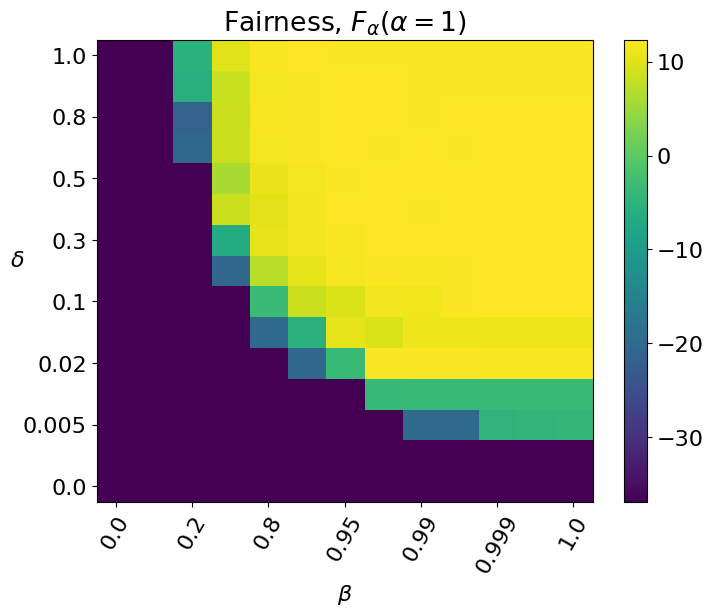}
        \includegraphics[width=0.495\linewidth]{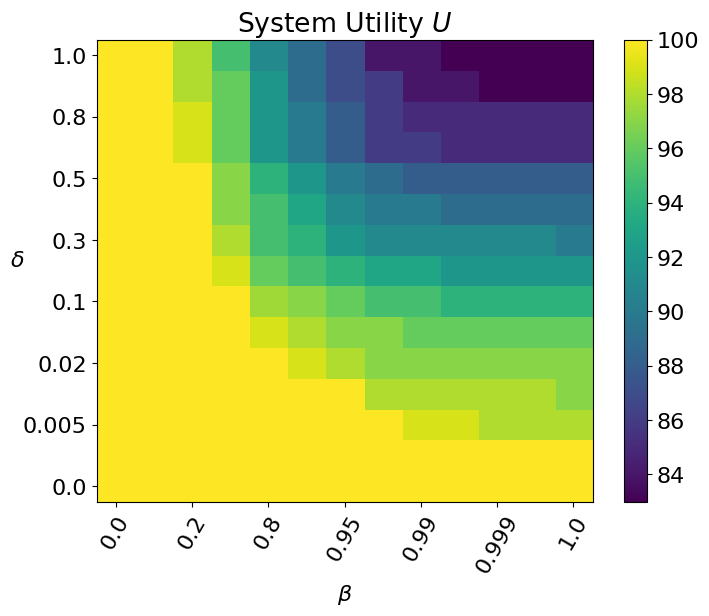}
        \caption{$\alpha$-fair results}
        \label{fig:alf_results_sub}
    \end{subfigure}
    \hfill
    \begin{subfigure}{1\linewidth}
        \centering
        \includegraphics[width=0.495\linewidth]{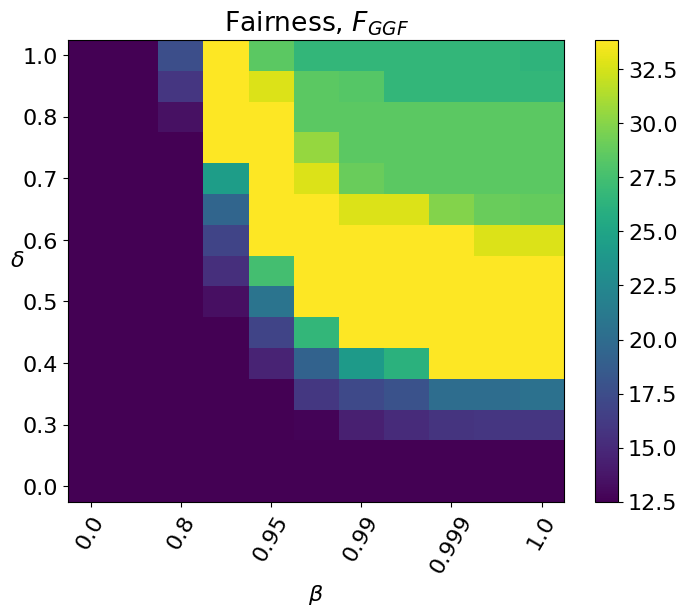}
        \includegraphics[width=0.495\linewidth]{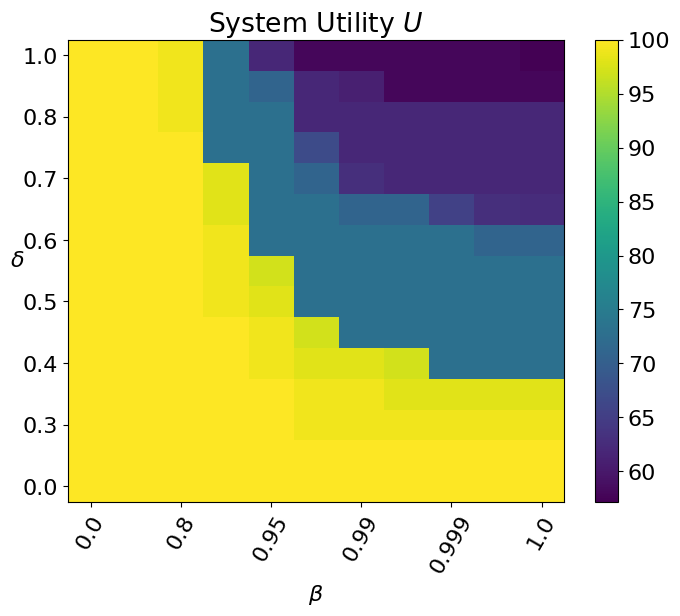}
        \caption{GGF results}
        \label{fig:ggf_results_sub}
    \end{subfigure}
    \caption{Fairness and utility for the Job Allocation environment as functions of $\beta$ and $\delta$.}
    \label{fig:subfig_results}
\end{figure}

\subsection{Job Allocation and Advantage Correction}

We evaluate our method in a challenging Job Allocation environment where 4 agents compete for a single job over 100 time steps. An agent occupying the job earns reward but must \texttt{forfeit} it (earning no reward for that step) to allow another agent to take over. This setup creates a tension between a greedy strategy (one agent gets ~100 utility), a simple turn-taking strategy (50 total utility, fairly split), and a hard-to-compute \textbf{oracle solution} that achieves 96 total utility (24 per agent). Our goal is to show that our method can discover this oracle solution via a simple hyperparameter search over $\beta$ and $\delta$, without needing to plan over the full time horizon.

Our approach relies on the advantage correction term, $\Delta Q_\text{adv}$. We test two \textbf{Efficient} fairness metrics ($\alpha$-fair and GGF), where the fairness gain $\Delta F$ is always positive. Consequently, without the correction term, an agent would never forfeit the job, and even a fully fairness-focused policy ($\beta=1$) would produce the same unfair outcome as a purely utilitarian one ($\beta=0$). The advantage correction solves this by reducing the value of actions that offer below-average fairness gains to agents that are already better off, thereby encouraging equitable turn-taking.

Our results, shown in Figure~\ref{fig:subfig_results}, confirm that this mechanism successfully balances utility and fairness:
\squishlist
    \item With the \textbf{$\alpha$-fair} metric (Figure~\ref{fig:alf_results_sub}), the policy achieves near-optimal performance with a moderate $\delta \approx 0.1$. As $\delta$ increases, fair behavior emerges at lower $\beta$ values.
    \item With the \textbf{GGF} metric (Figure~\ref{fig:ggf_results_sub}), the impact of advantage correction becomes noticeable around $\delta \approx 0.3$. This metric reveals a distinct band of high-utility, high-fairness solutions before the policy converges to a simple turn-taking strategy when both $\beta$ and $\delta$ are high. This difference occurs because the logarithmic nature of $\alpha$-fairness is less sensitive to utility changes once all agents are doing reasonably well.
\squishend

In both experiments, we observe that as $\delta$ increases, the transition from utilitarian to fair behavior happens at a smaller $\beta$. Crucially, our method identifies the optimal long-term strategy through a simple evaluation-only grid search, without any information about the time horizon. This demonstrates that the advantage correction term is a powerful and essential mechanism for achieving complex, far-sighted fairness without explicit planning.
\section{Discussion and Conclusion}

In this work, we introduced GIFF, a general and lightweight framework that integrates fairness into multi-agent resource allocation systems. By modifying pre-trained Q-values with a local fairness gain and a crucial counterfactual advantage correction, our method adjusts allocations to be more equitable without requiring any additional training. Our empirical evaluation showcased GIFF's power and versatility across diverse and challenging domains. In ridesharing, it consistently outperformed a strong, domain-specific baseline and demonstrated superior stability. In homelessness prevention, it proved its generality, adapting seamlessly with a Gini coefficient fairness metric. Finally, in the job allocation environment, the advantage correction mechanism was essential in discovering a near-oracle, long-horizon solution, demonstrating GIFF's power.

The practical strengths of GIFF are underpinned by a solid theoretical foundation. Its design is simple, with minimal computational overhead and just two interpretable hyperparameters ($\beta$,$\delta$) that allow for easy tuning of the fairness-utility trade-off at deployment. We formally proved that this is not merely a heuristic approach; GIFF's fairness surrogate is a principled lower bound on the true, realized fairness improvement for several canonical metrics. Furthermore, we showed that the fairness weight $\beta$ provides a monotonic guarantee, ensuring that increasing the emphasis on fairness predictably improves the surrogate objective.

% We also highlight some important avenues for future work. The current formulation assumes a centralized allocator and relies on the accuracy of Q-values, which could be a challenge in environments with highly stochastic or sparse rewards. Future research could focus on extending GIFF to decentralized systems and developing robustness to Q-value approximation errors. 
In conclusion, GIFF provides a practical, powerful, and principled bridge between the efficiency of reinforcement learning and the critical societal need for equity, representing a significant step toward creating multi-agent systems that are not only optimal but also just.

%%%%%%%%%%%%%%%%%%%%%%%%%%%%%%%%%%%%%%%%%%%%%%%%%%%%%%%%%%%%%%%%%%%%%%%%

%%% The acknowledgments section is defined using the "acks" environment
%%% (rather than an unnumbered section). The use of this environment 
%%% ensures the proper identification of the section in the article 
%%% metadata as well as the consistent spelling of the heading.

% \begin{acks}
% If you wish to include any acknowledgments in your paper (e.g., to 
% people or funding agencies), please do so using the `\texttt{acks}' 
% environment. Note that the text of your acknowledgments will be omitted
% if you compile your document with the `\texttt{anonymous}' option.
% \end{acks}

%%%%%%%%%%%%%%%%%%%%%%%%%%%%%%%%%%%%%%%%%%%%%%%%%%%%%%%%%%%%%%%%%%%%%%%%

%%% The next two lines define, first, the bibliography style to be 
%%% applied, and, second, the bibliography file to be used.

\bibliographystyle{ACM-Reference-Format} 
\bibliography{main}

%%%%%%%%%%%%%%%%%%%%%%%%%%%%%%%%%%%%%%%%%%%%%%%%%%%%%%%%%%%%%%%%%%%%%%%%
\clearpage
\newpage

\appendix

Here we present the supplementary material complementing the main paper, including full proofs for the theoretical results and extended discussion.

\section{Theoretical Results}
\label{sec-app:giff_theory_results}

In this section, we present the full proofs for the core theoretical properties of GIFF.  
Throughout this section, let $\Z=(z_1,\dots,z_n)\in\mathbb{R}^n$ denote the current payoff vector, and let $y=(y_1,\dots,y_n)\in\mathbb{R}^n_{\ge 0}$ denote the increments from a feasible allocation in the current round.  
For a fairness function $F:\mathbb{R}^n\to\mathbb{R}$ we define:
\begin{align}
\Delta_{\mathrm{joint}} &:= F(\Z+y)-F(\Z),\\
\Delta^{\mathrm{local}}_i &:= F(\Z+y_i e_i)-F(\Z)\\
S &:= \sum_{i=1}^n \Delta^{\mathrm{local}}_i
\end{align}
where $e_i$ is the $i$-th unit vector. We call $\Delta_{\mathrm{joint}}$ the \emph{realized fairness improvement}, $\Delta^{\mathrm{local}}_i$ the \emph{local fairness gain} for agent $i$, and $S$ the \emph{sum of local gains}, which we term as GIFF’s \textbf{surrogate}.  

Our proofs rely on the following assumption:
\begin{assumption}[Nonnegative increments]
\label{asmp-full:nonneg-y}
Any coordinate's change in utility is nonnegative for any allocation: $y_i \ge 0$ for all $i$.
\end{assumption}
\begin{assumption}[Q-value correctness]
\label{asmp-full:Q-corr}
For every agent $i$ and action $a$, the estimated Q-value equals the true utility increment:
\[
Q(o_i,a) \;=\; \Delta z_i(a),
\]
where $\Delta z_i(a)$ denotes the change in payoff $z_i$ that results from assigning action $a$ to agent $i$.
In other words, Q-values are perfectly accurate predictors of payoff increments.
\end{assumption}

Our results fall into three families:
\begin{enumerate}
  \item A \emph{Local–Gain Lower Bound}, showing that the surrogate never overstates realized fairness improvements.
  \item A \emph{Monotonicity in $\beta$}, showing that surrogate fairness is nondecreasing as the fairness weight $\beta$ increases.
  \item \emph{Slack bounds}, quantifying the gap between surrogate and realized fairness to turn surrogate guarantees into realized guarantees.
\end{enumerate}

\subsection{Local–Gain Lower Bound}
\label{subsec-app:theory-lg-lb}

Our first main result establishes that GIFF’s per-agent local gains form a certified lower bound on the realized joint fairness change. The result covers four canonical metrics: $\alpha$-fairness, negative variance, generalized Gini (GGF), and maximin.

% \input{ch7-GIFF-local-gain-bound}
% =========================================================
% Local–Gain Lower Bound: Theorem, Lemmas, and Proofs
% (uses \Z as the payoff vector; entries z_i)
% =========================================================

\begin{theorem}[Local–Gain Lower Bound]
\label{thm-full:local-gain-lb}
Let $\Z=(z_1,\dots,z_n)\in\mathbb{R}^n$ be a payoff vector and let
$y=(y_1,\dots,y_n)\in\mathbb{R}^n_{\ge 0}$ be a nonnegative increment vector.
For each fairness function $F$ below, the realized joint gain dominates the sum
of local gains:
\begin{align}
\Delta_{\mathrm{joint}}
&:= F(\Z+y)-F(\Z)
\;\;\ge\;\;
\sum_{i=1}^n \Delta^{\mathrm{local}}_i\\
\Delta^{\mathrm{local}}_i
&:= F(\Z+y_i e_i)-F(\Z)
\end{align}
under the stated assumptions:
\begin{enumerate}
\item $F=F_\alpha$ ($\alpha$-fairness): $\alpha\ge 0$ and $z_i,\ z_i+y_i$ lie in the domain
      of $U_\alpha$ (so $z_i>0$ and $z_i+y_i>0$ when $\alpha=1$).
\item $F=F_{\mathrm{var}}$ (negative variance): no further assumptions (beyond $y\ge 0$).
\item $F=F_{\mathrm{GGF}}$ (generalized Gini) with nonincreasing weights $w_1\ge\cdots\ge w_n$.
\item $F=F_{\min}$ (maximin): no further assumptions (beyond $y\ge 0$).
\end{enumerate}
Moreover, equality holds in (1); in (2) equality holds iff at most one $y_i>0$;
and in (4) equality conditions are given in Lemma~\ref{lem-full:maximin}.
\end{theorem}

\begin{lemma}[$\alpha$-fairness: exact additivity]
\label{lem-full:alpha}
Let $F_\alpha(\Z)=\sum_{i=1}^n U_\alpha(z_i)$ with
$U_\alpha(t)=\frac{t^{1-\alpha}}{1-\alpha}$ for $\alpha\neq 1$ (domain $t\ge 0$) and
$U_1(t)=\log t$ (domain $t>0$). If $z_i$ and $z_i+y_i$ are in the domain for all $i$, then
\[
F_\alpha(\Z+y)-F_\alpha(\Z)
=
\sum_{i=1}^n\!\bigl[F_\alpha(\Z+y_i e_i)-F_\alpha(\Z)\bigr].
\]
\end{lemma}
\begin{proof}
By separability,
$F_\alpha(\Z+y)-F_\alpha(\Z)=\sum_i[U_\alpha(z_i+y_i)-U_\alpha(z_i)]$,
and for a single-coordinate update
$F_\alpha(\Z+y_i e_i)-F_\alpha(\Z)=U_\alpha(z_i+y_i)-U_\alpha(z_i)$.
Summing over $i$ yields the identity.
\end{proof}

\begin{lemma}[Negative variance: nonnegative synergy]
\label{lem-full:variance}
Let $F_{\mathrm{var}}(\Z)=-\mathrm{Var}(\Z)$ with
$\mathrm{Var}(\Z)=\frac{1}{n}\sum_{i=1}^n (z_i-\mu)^2$ and $\mu=\frac{1}{n}\sum_i z_i$.
If $y\ge 0$, then
\[
F_{\mathrm{var}}(\Z+y)-F_{\mathrm{var}}(\Z)
\;\;\ge\;\;
\sum_{i=1}^n \bigl[F_{\mathrm{var}}(\Z+y_i e_i)-F_{\mathrm{var}}(\Z)\bigr],
\]
with equality iff at most one $y_i>0$.
\end{lemma}
\begin{proof}
Use $F_{\mathrm{var}}(\Z)=\mu^2-\frac{1}{n}\sum_i z_i^2$.
Let $S=\sum_i z_i$ and $Y=\sum_i y_i$. Then
\[
F_{\mathrm{var}}(\Z+y)-F_{\mathrm{var}}(\Z)
=\frac{2SY+Y^2}{n^2}-\frac{2}{n}\sum_i z_i y_i-\frac{1}{n}\sum_i y_i^2.
\]
For a single update $y_i$,
$F_{\mathrm{var}}(\Z+y_i e_i)-F_{\mathrm{var}}(\Z)
=\frac{2S y_i+y_i^2}{n^2}-\frac{2}{n}z_i y_i-\frac{1}{n}y_i^2$.
Summing over $i$ gives
\[
\sum_i \Delta_i^{\mathrm{local}}
=\frac{2SY}{n^2}-\frac{2}{n}\sum_i z_i y_i
+\Bigl(\frac{1}{n^2}-\frac{1}{n}\Bigr)\sum_i y_i^2.
\]
Subtracting yields
\(
\Delta_{\mathrm{joint}}-\sum_i \Delta_i^{\mathrm{local}}
=\frac{Y^2-\sum_i y_i^2}{n^2}
=\frac{2}{n^2}\sum_{i<j} y_i y_j\ge 0,
\)
with equality iff at most one $y_i>0$.
\end{proof}

\begin{lemma}[GGF: joint gain dominates locals]
\label{lem-full:ggf}
Let $F_{\mathrm{GGF}}(\Z)=\sum_{k=1}^n w_k\, z_{(k)}$ with nonincreasing weights
$w_1\ge\cdots\ge w_n$ (and $w_{n+1}:=0$), where $z_{(1)}\le\cdots\le z_{(n)}$ are the
sorted entries of $\Z$. If $y\in\mathbb{R}^n_{\ge 0}$, then
\[
F_{\mathrm{GGF}}(\Z+y)-F_{\mathrm{GGF}}(\Z)
\;\;\ge\;\;
\sum_{i=1}^n\!\bigl[F_{\mathrm{GGF}}(\Z+y_i e_i)-F_{\mathrm{GGF}}(\Z)\bigr].
\]
\end{lemma}
\begin{proof}
Use the equivalent form
\(
F_{\mathrm{GGF}}(\Z)=\sum_{k=1}^n v_k\, S_k(\Z)
\)
with
\(v_k:=w_k-w_{k+1}\ge 0\) and \(S_k(\Z)=\sum_{j=1}^k z_{(j)}\).
Since $v_k\ge 0$, it is enough to prove the claim for each $S_k$ and sum with weights $v_k$.

Fix $k$ and let $\tau$ be the $k$-th smallest value in $\Z$.
Set
\(
L=\{i:\ z_i<\tau\},\;
T=\{i:\ z_i=\tau\},\;
c=k-|L|.
\)
Any $k$-subset achieving $S_k(\Z)$ must include all of $L$ and exactly $c$ indices from $T$.

\emph{Local updates.}
For a single-coordinate increase $y_i\ge 0$,
\[
S_k(\Z+y_i e_i)-S_k(\Z)=
\begin{cases}
y_i, & i\in L,\\
y_i, & i\in T\ \text{and}\ c=|T|,\\
0,   & \text{otherwise}.
\end{cases}
\]
Therefore
\[
\sum_{i=1}^n\!\bigl[S_k(\Z+y_i e_i)-S_k(\Z)\bigr]
=\sum_{i\in L} y_i\;+\;\mathbf{1}_{\{c=|T|\}}\sum_{i\in T} y_i.
\tag{*}
\]

\emph{Joint update.}
For the full increment $y\ge 0$,
\[
\begin{aligned}
S_k(\Z+y)
&=\min_{|J|=k}\ \sum_{j\in J}(z_j+y_j)
\ \ge\
\min_{\substack{C\subseteq T\\ |C|=c}}
\left[\sum_{j\in L}(z_j+y_j)+\sum_{j\in C}(z_j+y_j)\right]\\
&= S_k(\Z)+\sum_{j\in L} y_j+\min_{\substack{C\subseteq T\\ |C|=c}}\sum_{j\in C} y_j.
\end{aligned}
\]
If $c=|T|$, the minimum over $C$ equals $\sum_{j\in T} y_j$; otherwise it is $\ge 0$.
Thus
\[
S_k(\Z+y)-S_k(\Z)\ \ge\ \sum_{j\in L} y_j\;+\;\mathbf{1}_{\{c=|T|\}}\sum_{j\in T} y_j,
\tag{**}
\]
which matches or exceeds the sum in (*). This proves the claim for $S_k$.
Multiplying by $v_k\ge 0$ and summing over $k$ yields the inequality for $F_{\mathrm{GGF}}$.
\end{proof}

\begin{lemma}[Maximin: joint never underestimates locals]
\label{lem-full:maximin}
Let $F_{\min}(\Z)=\min_i z_i$, let $m=\min_i z_i$, and let
$S=\{i:\,z_i=m\}$ be the set of minimizers. Let $\sigma$ be the second-smallest baseline
value (or $+\infty$ if none), and let
\[
\sigma' \;:=\; \min_{j\ne i^\star} (z_j + y_j)
\]
be the post-update second-smallest value when $|S|=1$ with $S=\{i^\star\}$.
If $y\ge 0$, then
\[
F_{\min}(\Z+y)-F_{\min}(\Z)
\;\;\ge\;\;
\sum_{i=1}^n\!\bigl[F_{\min}(\Z+y_i e_i)-F_{\min}(\Z)\bigr].
\]

\textbf{Equality conditions.}
\begin{itemize}
\item If $|S|\ge 2$ (multiple minima), equality holds iff $\min_{i\in S} y_i=0$.
      Otherwise the inequality is strict.
\item If $|S|=1$ with $S=\{i^\star\}$ (unique minimum), write
      \[
      \text{local sum} \;=\; \min\{y_{i^\star},\,\sigma-m\},
      \qquad
      \text{joint} \;=\; \min\{y_{i^\star},\,\sigma'-m\}.
      \]
      Equality holds iff either
      \begin{enumerate}
        \item $y_{i^\star} \le \sigma-m$ \ (then both sides equal $y_{i^\star}$), or
        \item $\sigma'=\sigma$ \ (then both sides equal $\min\{y_{i^\star},\,\sigma-m\}$).
      \end{enumerate}
      If $y_{i^\star}>\sigma-m$ and $\sigma'>\sigma$, the inequality is strict.
\end{itemize}
\end{lemma}

\begin{proof}
\emph{Local updates.}
If $i\notin S$, the minimum stays $m$, so the local change is $0$.
If $|S|\ge 2$ and $i\in S$, raising a single tied minimum still leaves some index at $m$,
so the local change is $0$. If $|S|=1$ with $S=\{i^\star\}$, then
\[
F_{\min}(\Z+y_{i^\star}e_{i^\star})-F_{\min}(\Z)
= \min\{\,m+y_{i^\star},\,\sigma\,\}-m
= \min\{\,y_{i^\star},\,\sigma-m\,\}.
\]
Summing over $i$ gives the stated “local sum.”

\emph{Joint update.}
If $|S|\ge 2$,
\[
F_{\min}(\Z+y)-F_{\min}(\Z)
= \min\{\,\min_{i\in S} y_i,\,\sigma-m\,\}\ \ge\ 0,
\]
so the inequality holds; equality iff $\min_{i\in S} y_i=0$.
If $|S|=1$ with $S=\{i^\star\}$,
\[
F_{\min}(\Z+y)-F_{\min}(\Z)
= \min\{\,y_{i^\star},\,\sigma'-m\,\}.
\]
Since $\sigma'\ge \sigma$ (because $y\ge 0$ on non-min entries), we have
\(
\min\{y_{i^\star},\sigma'-m\}\ge \min\{y_{i^\star},\sigma-m\}
\)
with equality exactly under (1) or (2) above.
\end{proof}

\bigskip
Now we are ready to prove the main result.

\begin{proof}[Proof of Theorem~\ref{thm-full:local-gain-lb}]
Each case follows from the corresponding lemma, which establishes
$\Delta_{\mathrm{joint}}\ge \sum_i \Delta^{\mathrm{local}}_i$ under the stated assumptions,
and provides the equality conditions.
\end{proof}

\noindent
This theorem implies that GIFF’s surrogate never \emph{over-promises}: the realized fairness gain is guaranteed to be at least as large as the surrogate, and is exact for $\alpha$-fairness. Equality conditions for variance and maximin are given in Lemmas~\ref{lem-full:variance} and \ref{lem-full:maximin}. See Subsection~\ref{subsec-app:theory-practical} for a detailed discussion of implications.

\subsection{Monotone Surrogate Fairness under $\beta$}
\label{subsec-app:theory-monotone}

Next, we show that as the fairness weight $\beta$ increases, the allocation’s sum of local fairness gains is nondecreasing. This gives a monotonic control knob: tuning $\beta$ cannot reduce surrogate fairness, and under uniqueness, any allocation switch yields a strict increase.

% \input{ch7-GIFF-monotone-surrogate}
% =========================================================
% Monotone Surrogate Fairness under GIFF: Theorem & Proof
% =========================================================

% Setup: baseline \Z, feasible allocations \mathcal A, utilities U and surrogate S.
% All \Delta F_i(\cdot) are computed against the SAME baseline \Z in this round.

\newcommand{\Sur}{S} % optional shorthand

\begin{theorem}[Monotone increase of the sum of local fairness gains]
\label{thm-full:monotone-surrogate}
Fix a decision round with baseline payoff vector $\Z\in\mathbb{R}^n$.
Let $\mathcal A$ be the (finite) set of feasible joint allocations $A=(a_1,\ldots,a_n)$.
For each $A\in\mathcal A$, define
\begin{align}
U(A) \;&:=\; \sum_{i=1}^n Q(o_i,a_i)\\
\Sur(A) \;&:=\; \sum_{i=1}^n \Delta F_i(a_i)\\
\Delta F_i(a_i) \;&:=\; F(\Z + y_i(a_i)\,e_i) - F(\Z)
\end{align}
where $y_i(a_i)\ge 0$ is the accounted increment added to $z_i$ if agent $i$ takes $a_i$,
and $e_i$ is the $i$-th unit vector. (All $\Delta F_i(\cdot)$ use the same baseline $\Z$.)

For $\beta\in[0,1)$, write $\theta=\beta/(1-\beta)$ and consider
\[
G_\theta(A)\;=\;U(A)\;+\;\theta\,\Sur(A).
\]
Let $A^*(\theta)\in\arg\max_{A\in\mathcal A} G_\theta(A)$ be any maximizer. Then for any
$0\le \theta_1<\theta_2<\infty$ and any choices
$A_1\in\arg\max G_{\theta_1}$, $A_2\in\arg\max G_{\theta_2}$,
\[
\Sur(A_2)\;\ge\;\Sur(A_1).
\]
Equivalently, as $\beta$ increases, the chosen allocation’s sum of local fairness gains is nondecreasing.
\end{theorem}

\begin{proof}
By optimality of $A_1$ at $\theta_1$ and $A_2$ at $\theta_2$,
\begin{align*}
U(A_1) + \theta_1 \Sur(A_1) \;\ge\; U(A_2) + \theta_1 \Sur(A_2), \tag{1}\\
U(A_2) + \theta_2 \Sur(A_2) \;\ge\; U(A_1) + \theta_2 \Sur(A_1). \tag{2}
\end{align*}
Subtract (1) from (2) to cancel the $U$ terms:
\[
(\theta_2-\theta_1)\,\big(\Sur(A_2)-\Sur(A_1)\big)\;\ge\;0.
\]
Since $\theta_2>\theta_1$, we conclude $\Sur(A_2)\ge \Sur(A_1)$.
\end{proof}

\begin{corollary}[Strict increase at a true switch under uniqueness]
\label{cor:strict-switch}
If the maximizer is unique at $\theta_1$ and at $\theta_2$, let
$A_1=A^*(\theta_1)$ and $A_2=A^*(\theta_2)$. If $A_1\neq A_2$, then
\[
\Sur(A_2) \;>\; \Sur(A_1).
\]
\end{corollary}

\begin{proof}
From Theorem~\ref{thm-full:monotone-surrogate}, $\Sur(A_2)\ge \Sur(A_1)$.
If equality held, then (1) and (2) above would force $U(A_1)=U(A_2)$ as well, so both
$A_1$ and $A_2$ would maximize $G_{\theta_1}$ and $G_{\theta_2}$—contradicting uniqueness.
\end{proof}

\paragraph{Scope, tie-breaking, and the $\beta\to 1$ endpoint.}
Theorem~\ref{thm-full:monotone-surrogate} requires only that (i) the feasible set $\mathcal A$ for the round is \emph{finite}, and (ii) all $\Delta F_i(\cdot)$ are computed against the \emph{same} baseline $\Z$ for that round. Changing the feasible set (e.g., constraints) or the baseline mid-sweep can break monotonicity; otherwise the result is agnostic to the choice of $F$ and to how $Q$ is obtained, provided $\Sur(A)$ is well defined. Deterministic tie-breaking is \emph{not} required for nondecreasing $\Sur(A^*(\theta))$, but if you want strict improvement at switches without assuming uniqueness, adopt a consistent rule such as: among $G_\theta$-maximizers, first maximize $\Sur(A)$, then $U(A)$. With this rule, any change in the selected allocation across $\theta_1<\theta_2$ implies $\Sur(A_2)>\Sur(A_1)$. Finally, interpreting $\beta\to 1^-$ as $\theta=\beta/(1-\beta)\to\infty$, the maximizers converge to $\arg\max_{A\in\mathcal A} \Sur(A)$, so the monotonicity statement extends continuously to the endpoint $\beta=1$.

% \begin{remark}[Endpoint $\beta=1$]
% Interpreting $\beta\to 1^-$ as $\theta\to\infty$, the maximizers converge to
% $\arg\max_{A\in\mathcal A}\Sur(A)$. The monotonicity statement extends to this endpoint by continuity.
% \end{remark}

% \paragraph{Remarks.}
% (i) Theorem~\ref{thm-full:monotone-surrogate} requires only that $\mathcal A$ be finite and that
% all $S(A)$ are computed against the same baseline $\mathbf Z$. It makes no assumptions on
% the specific fairness metric $F$ or the $Q$–function beyond well-definedness of $S(A)$.
% (ii) If desired, one can adopt a deterministic tie-break (e.g., pick among maximizers the
% one with largest $S$ and then largest $U$); the theorem still holds, and
% Corollary~\ref{cor:strict-switch} becomes immediate without separate uniqueness
% assumptions.

\noindent
This monotonicity applies to the surrogate $S(A)$, not necessarily the realized fairness $F(\Z)$. To connect the two, we require slack bounds.

\subsection{Slack Bounds Between Surrogate and Realized Fairness}
\label{subsec-app:theory-slack}

Define the \emph{slack} of an allocation as
\[
\mathrm{slack} \;:=\; \Delta_{\mathrm{joint}} - S \;\;\ge 0.
\]
By Theorem~\ref{thm-full:local-gain-lb}, $\mathrm{slack}\ge 0$ for the given fairness metrics.  
Here we provide per-metric exact formulas or bounds, so that we can certify not just lower bounds but also two-sided guarantees.

% \input{ch7-GIFF-surrogate-upper-bound}
% ch7-GIFF-surrogate-upper-bound.tex
% Lemmas giving per-metric formulas/bounds for slack = Delta_joint - S, with proofs.
% Notation matches the main text: baseline \Z, increments y, Y := \sum_i y_i, y_{\max} := \max_i y_i.

\begin{lemma}[Slack for $\alpha$-fairness]
\label{lem-full:slack-alf}
Let $F_\alpha(\Z)=\sum_i U_\alpha(z_i)$ with
$U_\alpha(t)=\frac{t^{1-\alpha}}{1-\alpha}$ for $\alpha\neq1$ (domain $t\ge 0$) and
$U_1(t)=\log t$ (domain $t>0$). Assume all arguments lie in the domain.
Then
\[
\Delta_{\mathrm{joint}} \;=\; S \quad\text{and}\quad \mathrm{slack}=0.
\]
\emph{Interpretation.} The surrogate is exact for any nonnegative increment profile.
\end{lemma}
\begin{proof}
By separability,
$F_\alpha(\Z+y)-F_\alpha(\Z)=\sum_i\!\bigl[U_\alpha(z_i+y_i)-U_\alpha(z_i)\bigr]$,
while for a single-coordinate update
$F_\alpha(\Z+y_i e_i)-F_\alpha(\Z)=U_\alpha(z_i+y_i)-U_\alpha(z_i)$.
Summing over $i$ gives $\Delta_{\mathrm{joint}}=S$.
\end{proof}

\begin{lemma}[Slack for negative variance]
\label{lem-full:slack-var}
Let $F_{\mathrm{var}}(\Z)=-\mathrm{Var}(\Z)$, and set $Y:=\sum_i y_i$.
Then
\[
\mathrm{slack}
\;=\; \frac{Y^2 - \sum_i y_i^2}{n^2}
\;=\; \frac{2}{n^2}\sum_{i<j} y_i y_j
\;\;\in\;\;\Big[\,0,\;\frac{Y^2}{n^2}\Big(1-\tfrac{1}{m}\Big)\,\Big],
\]
where $m:=|\{i:\,y_i>0\}|$.
\emph{Interpretation.} The gap is a pure “synergy” term that grows when gains are
spread across more agents; it vanishes when at most one coordinate increases.
\end{lemma}
\begin{proof}
Using $F_{\mathrm{var}}(\Z)=\mu^2-\frac{1}{n}\sum_i z_i^2$ with $\mu=\frac{1}{n}\sum_i z_i$,
the calculation in Lemma~\ref{lem-full:variance} gives
\[
\Delta_{\mathrm{joint}}-\sum_i \Delta^{\mathrm{local}}_i
=\frac{Y^2-\sum_i y_i^2}{n^2}
=\frac{2}{n^2}\sum_{i<j}y_i y_j \;\ge 0.
\]
For the upper bound, by Cauchy–Schwarz, with $m$ positive entries,
$\sum_i y_i^2\ge Y^2/m$, hence $Y^2-\sum_i y_i^2\le Y^2(1-1/m)$.
\end{proof}

% \begin{lemma}[Slack for GGF (nonincreasing weights)]
% \label{lem-full:slack-ggf}
% Let $F_{\mathrm{GGF}}(\Z)=\sum_{k=1}^n w_k\, z_{(k)}$ with nonincreasing weights
% $w_1\ge\cdots\ge w_n$ and $w_{n+1}:=0$, and let $Y:=\sum_i y_i$.
% Then
% \[
% 0 \;\le\; \mathrm{slack} \;\le\; \Delta_{\mathrm{joint}}.
% \]
% If the baseline order of $\Z$ is strict and preserved after the update, writing
% $y_{(1)}\le\cdots\le y_{(n)}$ aligned with that order,
% \[
% \Delta_{\mathrm{joint}} \;=\; \sum_{k=1}^n w_k\, y_{(k)} \;=\; S,
% \qquad\text{so }\mathrm{slack}=0.
% \]
% \emph{Interpretation.} Exactness holds under no rank crossings; otherwise slack can be positive.
% \end{lemma}
% \begin{proof}
% Nonnegativity of $\mathrm{slack}$ follows from Theorem~\ref{thm-full:local-gain-lb}. If the baseline
% order is preserved, each rank-$k$ index stays at rank $k$, so
% $\Delta_{\mathrm{joint}}=\sum_k w_k y_{(k)}=S$.
% \end{proof}

% \begin{corollary}[$y_{\max}$-cap for GGF]
% \label{cor:ggf-ymax}
% Under the setting of Lemma~\ref{lem-full:slack-ggf}, let $m:=|\{i:\,y_i>0\}|$ and
% $y_{\max}:=\max_i y_i$. Define
% \[
% q:=\min\!\Big(m,\,\big\lfloor Y/y_{\max}\big\rfloor\Big),\qquad
% r:=Y-q\,y_{\max}\in[0,y_{\max}).
% \]
% Then
% \[
% \Delta_{\mathrm{joint}}
% \;\le\;
% y_{\max}\sum_{k=1}^{q} w_k \;+\; r\,w_{q+1}
% \;\;\le\;\; y_{\max}\sum_{k=1}^{m} w_k,
% \]
% and hence $\mathrm{slack}\le \big[y_{\max}\sum_{k=1}^{q} w_k + r\,w_{q+1}\big]-S$.
% \emph{Interpretation.} This bound is typically much tighter than crude Lipschitz bounds when gains
% are spread across agents.
% \end{corollary}

\begin{lemma}[Slack for GGF (nonincreasing weights)]
\label{lem-full:slack-ggf}
Let $F_{\mathrm{GGF}}(\Z)=\sum_{k=1}^n w_k\, z_{(k)}$ with nonincreasing weights $w_1\ge\cdots\ge w_n$ and $w_{n+1}:=0$, and let the increments $y$ be nonnegative. Set $Y:=\sum_i y_i$, $m:=|\{i:\,y_i>0\}|$, and $y_{\max}:=\max_i y_i$. Define
\[
q:=\min\!\Big(m,\,\big\lfloor Y/y_{\max}\big\rfloor\Big),\qquad
r:=Y-q\,y_{\max}\in[0,y_{\max}).
\]
Then
\[
\Delta_{\mathrm{joint}}
\;\le\;
y_{\max}\sum_{k=1}^{q} w_k \;+\; r\,w_{q+1}
\]
and hence,
\[
\mathrm{slack}
\;\le\;
\Big(y_{\max}\sum_{k=1}^{q} w_k + r\,w_{q+1}\Big) - S.
\]
If the baseline order of $\Z$ is strict and preserved after the update, writing
$y_{(1)}\le\cdots\le y_{(n)}$ aligned with that order,
\[
\Delta_{\mathrm{joint}} \;=\; \sum_{k=1}^n w_k\, y_{(k)} \;=\; S,
\quad\text{so }\mathrm{slack}=0.
\]
\emph{Interpretation.} The $y_{\max}$-cap is a simple data-dependent envelope (depending only on $m$, $Y$, $y_{\max}$ and the weights); it is typically tight when mass is spread.
\end{lemma}
\begin{proof}
Let $y_{[1]}\ge\cdots\ge y_{[n]}$ be the increments sorted descending. By rearrangement,
\[
\Delta_{\mathrm{joint}}
\;\le\; \sum_{k=1}^n w_k\,y_{[k]}.
\]
Maximizing the right-hand side under $0\le y_{[k]}\le y_{\max}$, $\sum_k y_{[k]}=Y$, and at most $m$ positives fills the top $q$ slots with $y_{\max}$ and places the remainder $r$ in slot $q{+}1$, yielding $y_{\max}\sum_{k=1}^{q} w_k + r\,w_{q+1}$. Since $\mathrm{slack}=\Delta_{\mathrm{joint}}-S$, the stated slack bound follows. In the no-rank-crossing regime, each baseline rank-$k$ index remains at rank $k$, so the joint change equals $\sum_k w_k y_{(k)}=S$.
\end{proof}

\begin{lemma}[Slack for maximin]
\label{lem-full:slack-min}
Let $F_{\min}(\Z)=\min_i z_i$. Denote $m^\star:=\min_i z_i$ and
$S:=\{i:\,z_i=m^\star\}$. Let $\sigma:=\min_{j\notin S} z_j$ (or $+\infty$ if $S=\{1,\dots,n\}$),
and define $y_{\max}:=\max_i y_i$.
\begin{enumerate}
\item If $|S|\ge 2$,
\[
\mathrm{slack}
\;=\; \min\!\big\{\min_{i\in S} y_i,\; \sigma - m^\star\big\}
\;\in\;[0,\,y_{\max}],
\]
with $\mathrm{slack}=0$ iff $\min_{i\in S}y_i=0$.
\item If $|S|=1$ with $S=\{i^\star\}$, let
$\sigma':=\min_{j\neq i^\star}(z_j+y_j) \ge \sigma$. Then
% \[
% \mathrm{slack}
% \;=\; \min\{\,y_{i^\star},\,\sigma'-m^\star\,\}
% \;-\; \min\{\,y_{i^\star},\,\sigma-m^\star\,\}
% \;\in\; \big[\,0,\;\min\{y_{i^\star},\,\sigma'-\sigma\}\,\big].
% \]
\[
\begin{split}
\mathrm{slack} ={}& \min\{\,y_{i^\star},\,\sigma'-m^\star\,\} - \min\{\,y_{i^\star},\,\sigma-m^\star\,\} \\
              & \in \big[\,0,\;\min\{y_{i^\star},\,\sigma'-\sigma\}\,\big].
\end{split}
\]
In particular, $\mathrm{slack}=0$ if either $y_{i^\star}\le\sigma-m^\star$ or $\sigma'=\sigma$.
\end{enumerate}
\end{lemma}
\begin{proof}
Let $\Delta_{\rm joint}=F_{\min}(\Z+y)-F_{\min}(\Z)$ and
$\Delta^{\rm local}_i=F_{\min}(\Z+y_i e_i)-F_{\min}(\Z)$.

(1) If $|S|\ge 2$, every single-coordinate update leaves some entry at $m^\star$, so
$\Delta^{\rm local}_i=0$ for all $i$ and $S=0$. The new minimum after the joint update is
$\min\{\,m^\star+\min_{i\in S}y_i,\ \sigma\,\}$, hence
$\Delta_{\rm joint}=\min\{\min_{i\in S}y_i,\ \sigma-m^\star\}$ and the stated bounds follow.

(2) If $|S|=1$ with $S=\{i^\star\}$, then
$\Delta^{\rm local}_{i^\star}=\min\{y_{i^\star},\sigma-m^\star\}$ and
$\Delta^{\rm local}_i=0$ for $i\neq i^\star$, so $S=\min\{y_{i^\star},\sigma-m^\star\}$.
After the joint update, the second-smallest value becomes
$\sigma'=\min_{j\neq i^\star}(z_j+y_j)\ge\sigma$, hence
$\Delta_{\rm joint}=\min\{y_{i^\star},\sigma'-m^\star\}$. Subtracting gives the claim and bounds.
\end{proof}

% \paragraph{Practical remarks.}
% (i) The exact slack for variance is $O(n)$ to compute; for GGF the bound reuses the baseline order structure and needs only selecting the $c_k$ smallest $y$’s in each tie set $T_k$. (ii) These two-sided bounds turn the surrogate $S$ into a \emph{certificate}: a guaranteed minimum fairness improvement and a data-dependent cap on how much the surrogate might underestimate realized change. (iii) In many applications (no ties at thresholds, unique minimum, sparse increments) the slack is often zero or small.

\noindent
These results allow us to strengthen the monotonicity guarantee: if $S$ increases by more than an upper bound on $\mathrm{slack}$, then realized fairness must also strictly increase. For $\alpha$-fairness, this is immediate because $\mathrm{slack}=0$.

\subsection{Corollary: Strict Realized Fairness at $\beta$-Driven Switches}
\label{subsec-app:theory-strict}

Combining Theorem~\ref{thm-full:monotone-surrogate} with the slack bounds above yields a useful corollary: when reallocations occur as $\beta$ increases, and the increase in surrogate fairness exceeds the slack bound of the previous allocation, realized fairness must strictly improve.  

This gives a precise, testable condition: 
\begin{itemize}
  \item For $\alpha$-fairness, every true switch increases realized fairness.  
  \item For variance, a computable quadratic bound applies (Lemma~\ref{lem-full:slack-var}).  
  \item For GGF, a $y_{\max}$-cap yields a data-dependent bound, with slack vanishing when no rank crossings occur (Lemma~\ref{lem-full:slack-ggf}).  
  \item For maximin, the gap depends on whether the minimum is unique (Lemma~\ref{lem-full:slack-min}).  
\end{itemize}

\medskip
In practice, this means $\beta$ can be tuned with confidence that surrogate monotonicity is preserved, and—in many metrics—realized fairness improves as well.

\subsection{Practical Implications of the Theoretical Results}
\label{subsec-app:theory-practical}

The preceding theorems are not only of theoretical interest but also provide practical
tools for deploying GIFF in real systems:

\paragraph{Certified lower bounds.}
The Local–Gain Lower Bound guarantees that GIFF’s surrogate never overstates realized
fairness improvement. This turns the surrogate into a safe proxy: if an allocation is
predicted to achieve at least $\varepsilon$ fairness gain, the realized gain is
provably $\ge \varepsilon$. For $\alpha$-fairness, the surrogate is exact.

\paragraph{Safe tuning of $\beta$.}
The monotonicity theorem ensures that increasing the fairness weight $\beta$ cannot
reduce surrogate fairness. Practitioners can therefore adjust $\beta$ to explore
efficiency–fairness trade-offs without fear of hidden regressions.

\paragraph{Two-sided certificates.}
Slack bounds provide upper bounds on how much realized fairness can exceed the surrogate.
Together with the lower bound, this yields a sandwich:
\[
S \;\le\; \Delta_{\mathrm{joint}} \;\le\; S+\mathrm{slack}_{\max}.
\]
This makes GIFF auditable: both the minimum guaranteed improvement and the potential
gap are known at each round.

\paragraph{Fairness floors as constraints.}
Because $S$ is conservative, one can impose hard constraints of the form
$S \ge \varepsilon$ in the allocation ILP, guaranteeing that realized fairness
improvement is at least $\varepsilon$ each round.

\paragraph{Telescoping guarantees.}
Summing the Local–Gain Lower Bound over time gives a cumulative guarantee:
\[
F(\Z_T)-F(\Z_0) \;\;\ge\;\; \sum_{t=0}^{T-1} S^{(t)}.
\]
This yields an auditable trajectory of fairness progress over a deployment horizon.

\paragraph{Monitoring and debugging.}
The bounds enable runtime checks. If the observed joint change $\Delta_{\mathrm{joint}}$
ever falls below the computed $S$, then assumptions (such as nonnegative increments) have been violated, or an implementation error is present.

\medskip
\noindent
In short, the theorems make GIFF \emph{deployable}: the system’s fairness behavior
becomes predictable, auditable, and tunable in ways that are both theoretically grounded
and operationally meaningful.

\section{Simple Incentives - Extended (SI-X) Baseline}
\begin{table*}[t]
\centering
\caption{The effect of the SI-X fairness incentive $F(a)$ for $\beta>0$.}
\label{tab:six_effects}
\begin{tabular}{|c|c|c|}
\hline
 & \textbf{Better-off group} & \textbf{Worse-off group} \\
 & ($\bar{z} - z_{G(h)} > 0$) & ($\bar{z} - z_{G(h)} < 0$) \\ \hline
\textbf{Bad action} & Discourage & Heavily discourage \\
($\Pr(h,a) - z_{G(h)} > 0$) & (Negative incentive) & (Positive incentive, but bad action) \\ \hline
\textbf{Good action} & Lightly encourage & Heavily encourage \\
($\Pr(h,a) - z_{G(h)} < 0$) & (Negative incentive) & (Positive incentive) \\ \hline
\end{tabular}
\end{table*}

To provide a strong point of comparison for our homelessness prevention experiments, we developed a competitive baseline by directly adapting the \textbf{Simple Incentives (SI)} framework~\citep{SI_kumar2023}. The original method was designed to mitigate fairness issues in ridesharing by moving allocations closer to statistical parity. We apply its core variance-minimization logic to the homelessness domain, leading to the variant we term SI-X.

The central idea of the original SI framework is to treat the variance of outcomes across all groups, $\text{var}(\mathbf{Z})$, as a proxy for unfairness. At each step, the goal is to make an allocation that takes a gradient step toward minimizing this variance. The original paper derives a modified score function, $s'$, by subtracting the gradient of the variance from the original utility score, $s(i,a)$:
\begin{align}
    s'(i,a) = s(i,a) - \lambda \, \frac{\partial}{\partial \mathcal{A}} \var(\textbf{Z}) = s(i,a) + \frac{2}{|\textbf{Z}|} \lambda \, \sum_{z_j \in \textbf{Z}} (\bar{z} - z_j) \frac{\partial z_j}{\partial \mathcal{A}}
\end{align}
where $\lambda$ is a hyperparameter, $z_j$ is the historical average outcome for group $j$, and $\bar{z}$ is the average across all groups. The second term is the fairness \emph{incentive}: for each group, its disparity from the mean, $(\bar{z} - z_j)$, is weighted by the impact of the allocation on its outcome, $\frac{\partial z_j}{\partial \mathcal{A}}$.

To apply this framework to the homelessness domain, we make two simplifying assumptions in line with the assumptions made by the original work:
\begin{enumerate}
    \item \textbf{Simplification}: For an action $a$ involving a single household $h$, the sum over all groups collapses to just the term for that household's group, $G(h)$, since no other group's metric is affected.
    \item \textbf{Approximation}: The complex partial derivative, $\frac{\partial z_{G(h)}}{\partial a}$, is approximated with a simple, intuitive heuristic: the difference between the action's re-entry probability and the group's historical average, $\Pr(h,a) - z_{G(h)}$.
\end{enumerate}

These steps yield the final form for the fairness incentive, $F(a)$, used in our experiments. The modified utility for an action, $s'(a)$, becomes:
\begin{align}
    s'(a) &= Q(a) + \beta F(a) \\
    \text{where } F(a) &= \underbrace{(\bar{z} - z_{G(h)})}_{\text{group advantage}} \times \underbrace{(\Pr(h,a) - z_{G(h)})}_{\text{action advantage}}
\end{align}
and $Q(a) = -\Pr(h,a)$. The two components of this incentive have a clear intuition:
\begin{itemize}
    \item \textbf{Group Advantage}: The term $(\bar{z} - z_{G(h)})$ is positive for better-off groups (where $z_{G(h)} < \bar{z}$) and negative for worse-off groups.
    \item \textbf{Action Advantage}: The term $(\Pr(h,a) - z_{G(h)})$ is negative for "good" actions (lower re-entry risk than the group average) and positive for "bad" actions.
\end{itemize}

The product of these two terms, summarized in Table \ref{tab:six_effects}, creates a targeted incentive structure that directly reflects the original framework's goal of reducing variance by penalizing actions that increase disparity and rewarding those that reduce it.

\section{GIFF Algorithm}

Here we provide general algorithms for GIFF. These can be specialized for specific use cases and improved with practical modifications, like vectorization or short-circuiting when $\beta$ or $\delta$ is zero. The main algorithm is provided in Algorithm~\ref{alg:giff_general_main}, with the advantage correction process outlined in Algorithm~\ref{alg:advantage_correction_general}.

\begin{algorithm}[H]
\caption{General Incentives-based Framework for Fairness (GIFF)}
\label{alg:giff_general_main}
\begin{algorithmic}[1]
\Require Set of agents $\alpha$, Current payoff vector $\mathbf{Z}$, Q-values $Q(i,a)$, Fairness function $F$, Hyperparameters $\beta, \delta$
\State Initialize modified utility matrix $Q_{GIFF}$
\For{each agent $i \in \alpha$}
    \For{each available action $a \in A_i$}
        \State // \textit{1. Compute Local Fairness Gain}
        \State $\Delta F(a) \gets \text{ComputeFairnessGain}(i, a, Q(i,a), \mathbf{Z}, F)$
        \State
        \State // \textit{2. Compute Advantage Correction (see Algorithm \ref{alg:advantage_correction_general})}
        \State $\Delta Q_{\text{adv}}(a) \gets \text{AdvantageCorrection}(i, a, \{Q\}, \mathbf{Z}, F)$
        \State
        \State // \textit{3. Combine into GIFF-modified Utility}
        \State $Q_{f}(a) \gets \Delta F(a) + \delta \cdot \Delta Q_{\text{adv}}(a)$
        \State $Q_{GIFF}(i,a) \gets (1-\beta) \cdot Q(i,a) + \beta \cdot Q_{f}(a)$
    \EndFor
\EndFor
\State // \textit{4. Solve for Final Allocation}
\State $\mathcal{A}^* \gets \text{SOLVE\_ALLOCATION}(Q_{GIFF})$
\State \Return The fair allocation $\mathcal{A}^*$
\end{algorithmic}
\end{algorithm}

\begin{algorithm}[H]
\caption{Counterfactual Advantage Correction}
\label{alg:advantage_correction_general}
\begin{algorithmic}[1]
\Require Agent $i$, Action $a$, Q-values $\{Q(j,a)\}$, Current payoff vector $\mathbf{Z}$, Fairness function $F$
\Function{AdvantageCorrection}{$i, a, \{Q\}, \mathbf{Z}, F$}
    \State // \textit{1. Compute agent i's local fairness gain}
    \State $\Delta F(a) \gets \text{ComputeFairnessGain}(i, a, Q(i,a), \mathbf{Z}, F)$
    \State
    \State // \textit{2. Compute average counterfactual fairness gain}
    \State Initialize list $\mathcal{F}_{\text{cf}}$
    \For{each agent $j \neq i$ that can take action $a$}
        \State $\Delta F^{(j)} \gets \text{ComputeFairnessGain}(j, a, Q(j,a), \mathbf{Z}, F)$
        \State Append $\Delta F^{(j)}$ to $\mathcal{F}_{\text{cf}}$
    \EndFor
    \State $\Delta F_{\text{avg}}(a) \gets \text{mean}(\mathcal{F}_{\text{cf}})$
    \State
    \State // \textit{3. Compute final correction term}
    \State $F_{\text{adv}}(a) \gets \Delta F(a) - \Delta F_{\text{avg}}(a)$
    \State $\Delta Q(a) \gets Q(i, a) - \min_{a' \in A_i} Q(i, a')$
    \State $\Delta Q_{\text{adv}}(a) \gets F_{\text{adv}}(a) \cdot \Delta Q(a)$
    \State \Return $\Delta Q_{\text{adv}}(a)$
\EndFunction
\end{algorithmic}
\end{algorithm}

\section{Implementation Details for the Ridesharing Experiment}

We used the code provided by the authors of SI~\cite{SI_kumar2023} directly, modifying the parts where the score is modified by SI, instead using GIFF's approach.
The ridesharing experiments were conducted in a simulation environment based on prior work~\citep{shah2020neural, SI_kumar2023}.

\begin{itemize}
    \item \textbf{Environment}: The simulation models a fleet of 1,000 vehicles in Manhattan using a real-world NYC dataset. The experiments focus on the morning rush hour from 8 am to 12 pm.
    
    \item \textbf{Q-Values}: We use the same pre-trained Q-value models from the baseline work to ensure a direct and fair comparison. A central allocator uses these Q-values to assign passengers to drivers.
    
    \item \textbf{Fairness Metrics}: We evaluate fairness for both passengers and drivers, using the same definitions as the baseline method.
    \begin{itemize}
        \item \textbf{Passenger Fairness}: Passengers are grouped by their origin-destination neighborhood pair. The payoff vector, $\mathbf{Z}_p$, consists of the service rate for each group. Fairness is measured as the negative variance, $-var(\mathbf{Z}_p)$.
        
        \item \textbf{Driver Fairness}: Each driver is their own group. The payoff vector, $\mathbf{Z}_d$, is the cumulative number of trips assigned to each driver. Fairness is measured as $-var(\mathbf{Z}_d)$.
    \end{itemize}
\end{itemize}

\section{Implementation Details for the Homelessness Experiment}
\label{app:homelessness_details}

This section provides additional details on the experimental setup for the homelessness prevention domain, complementing the description in the main text.

\subsection{Data and Preprocessing}
The experiment utilizes two primary datasets, based on the work of \citet{kube2019allocating}. The first contains counterfactual re-entry probabilities for 13,940 households. These probabilities were generated using Bayesian Additive Regression Trees (BART) and estimate the likelihood of a household re-entering the homelessness system if assigned to one of four interventions: \textbf{Emergency Shelter (ES)}, \textbf{Transitional Housing (TH)}, \textbf{Rapid Re-housing (RRH)}, and \textbf{Homelessness Prevention (Prev)}. The second dataset, contains demographic and background features for each household. These two datasets were merged using the household identifier.

To define the fairness groups for our analysis, we automatically selected relevant features from the household data. We filtered for categorical features with at least two and at most twenty unique values. To ensure statistical stability, we further required that each unique value (i.e., each subgroup) within a feature be associated with at least 50 households in the dataset. This process yielded \textbf{38 distinct features} (e.g., race, family size, disability status) that were used to define demographic groups in 38 independent experimental runs.

\subsection{Temporal Simulation Setup}
To simulate a realistic, dynamic allocation process, we structured the experiment into discrete time windows. Using the \texttt{EntryDate} for each household, we divided the dataset into sequential, non-overlapping windows of \textbf{30 days}. All households entering the system within a given 30-day period were considered for allocation simultaneously at the end of that window.

\subsection{Resource Allocation and Constraints}
The experiment was conducted in a \textbf{constrained allocation setting}, where the total number of available slots for each intervention was fixed over the entire simulation. These totals were derived from the historical assignments in the original dataset:
\begin{itemize}
    \item \textbf{Prevention (Prev)}: 6202 slots
    \item \textbf{Emergency Shelter (ES)}: 4441 slots
    \item \textbf{Transitional Housing (TH)}: 2451 slots
    \item \textbf{Rapid Re-housing (RRH)}: 846 slots
\end{itemize}

The total available slots were not divided evenly across time. Instead, for each 30-day window, a number of slots for each intervention was made available proportionally. The number of slots in a given window was scaled based on the number of households arriving in that window relative to the total number of households. This ensures that resource availability realistically matches demand over time. Within each time window, an ILP solver was used to find the optimal assignment of households to the available intervention slots, based on the (potentially fairness-modified) utility scores.

\subsection{Fairness Metric and Evaluation Protocol}
The fairness objective in this experiment was to minimize the disparity in average re-entry probabilities across demographic groups. We used the \textbf{Gini index} as our measure of inequality. Since our framework is formulated as a maximization problem, the specific fairness function provided to the methods was the \textbf{negative Gini index} (\texttt{fairness\_function = lambda x: -gini(x)}).

The full experiment consisted of a loop over the 38 selected demographic features. For each feature, we ran the entire temporal simulation for both the GIFF and SI-X methods. For each method, we performed a sweep across a range of the fairness hyperparameter ($\beta$) to trace the trade-off between overall efficiency (total re-entry probability) and fairness (Gini index). The results from all 38 runs were then aggregated to produce the final distributions shown in the main paper.

\section{Practical Notes}
\subsection{A Note on Computational Overhead:}
Assuming that evaluating the fairness function on a payoff vector takes constant time, computing the GIFF-modified Q-value for one agent involves two main steps for each available action. First, the fairness gain for an action is computed in constant time. Without advantage correction, this is the same computational complexity as it takes to enumerate all Q-values. Second, the advantage correction term requires evaluating the counterfactual fairness gain for up to $n$ candidate agents, resulting in $O(n)$ time per action. With an average of $m$ actions available per agent, the overall computational overhead for calculating the modified Q-values for one agent is $O(m \cdot n)$. The total time for all agents, then, is $O(m \cdot n^2)$, which is much smaller than evaluating fairness over the total combinatorial search space of $O(m^n)$ joint actions.

\subsection{Clarification of the Q-value Assumption in Theoretical Analysis}

Our theoretical results in Appendix \ref{sec-app:giff_theory_results}, particularly the local-gain lower bound (Theorem \ref{thm-full:local-gain-lb}) and the monotonicity guarantee (Theorem \ref{thm-full:monotone-surrogate}), are derived under the idealizing assumption of Q-value correctness (Assumption \ref{asmp-full:Q-corr}). Specifically, we assume that the Q-value for an agent's action, $Q(o_i, a)$, represents the exact, single-step utility increment, $y_i$, that the agent will realize from that allocation in the current timestep.

We acknowledge that this is a simplification. In practice, Q-values are learned estimates of the long-term, discounted sum of future rewards and are subject to approximation errors, especially when function approximators like deep neural networks are used. They are not equivalent to the true, immediate utility gain.
However, if the Q-values are accurate representations of long-term value attained by the agent, the fairness improvement will still be correct, even though it will not be realized within the single time step.

Additionally, this kind of assumption is standard for the theoretical analysis of mechanisms built on top of reinforcement learning. Its purpose is to isolate and analyze the properties of the GIFF framework itself, disentangled from the separate and complex issue of Q-value estimation error. By assuming perfect Q-values as inputs, we can prove that the GIFF mechanism for translating utility into fairness is principled—that its surrogate is a valid lower bound and that its tuning parameter $\beta$ behaves predictably. The strong performance of GIFF in our diverse empirical evaluations, which use Q-values learned in complex and stochastic environments (or which appear from black-box sources), suggests that these desirable theoretical properties are robust enough to hold in practice even when this assumption is relaxed.

\section{Broader Impact Statement}
%\label{sec:broaderImpact}
Methods to approach fairness often come at a cost to some participants' utility.
Resource allocation using GIFF has the potential to affect real-world stakeholders. However, users should exercise caution, as reliance on inaccurate learned Q-values or poor fairness metrics may propagate existing biases in the data
, and misapplication in high-stakes environments could lead to unforeseen inequitable outcomes.

\end{document}